\documentclass[10pt,reqno,oneside]{amsart}

\usepackage[T1]{fontenc}
\usepackage[utf8]{inputenc}
\usepackage[british]{babel}
\usepackage{amsmath,amssymb,amsthm}
\usepackage{geometry}
\usepackage{microtype}
\usepackage{tikz}
\usetikzlibrary{calc}
\usepackage{xcolor}
\usepackage{hyperref}
\hypersetup{colorlinks=true,linkcolor=blue!55!black,citecolor=blue!55!black,
            urlcolor=blue!55!black}

\theoremstyle{plain}
\newtheorem{theorem}{Theorem}[section]

\newtheorem{proposition}[theorem]{Proposition}

\theoremstyle{definition}

\theoremstyle{remark}

\newcommand{\Sph}{\mathbb{S}}
\newcommand{\Reals}{\mathbb{R}}
\def\d{{\rm d}}
\def\beq{\begin{equation}}
\def\eeq{\end{equation}}
\newcommand{\iiota}{\dot\iota}
\newcommand{\posted}{\overset{!}{=}}

\numberwithin{equation}{section}

\begin{document}

\title{The kinematic structures and the inertial geometry of a moving charge}

%% Byline: that of arXiv:2608.21642, on his instruction of 2026-09-02 --- one author.
\author{C.~S. L\'opez-Monsalvo}
\address{Departamento de Ciencias B\'asicas, Universidad Aut\'onoma Metropolitana --
  Azcapotzalco, Avenida San Pablo 420, Colonia Nueva El Rosario, Azcapotzalco 02128,
  Ciudad de M\'exico, Mexico}
\email{cslm@azc.uam.mx}

\date{\today}

\begin{abstract}
We ask how much of the geometry a charged particle moves in is fixed by its motion, and how much a particle must bring. We ask of a symplectic structure only that it relate velocity to momentum as Hamilton's equations do, and we ask it of every energy at once. In particular, we show that the structures meeting that demand are the canonical one and its twists by a closed two-form of the base. A field provides the two-form, a particle the multiplier before it, which we identify constitutively with its charge. Thus, a single energy governs a family of structures, and each particle takes the one its charge fixes. We then ask what a particle must bring to be given a momentum, and we show that the degree of that map settles the degree at which a field enters Newton's Second Law. An antisymmetric bilinear form returns no Lorentz force, whilst a Randers metric returns one --- a length whose difference from a Riemannian one is linear in the velocity. Moreover, we find that metric already within the twisted structure, as its primitive over a level of the free energy, and its law of transport to be nonlinear, no affine connection being known to serve. Under an indefinite signature the length parts from the dynamics, and the extremals turn from shortest to longest. On the round sphere a monopole flux leaves no such metric, whilst the transport remains and prequantisation, given a unit of action, restricts the charge to a lattice. In this manner, we conclude that each charge-to-mass ratio receives a geometry of its own, so that by a functionalist criterion none of them is the spacetime of a charged particle.
\end{abstract}

\maketitle

%%============================================================
\section{Introduction}\label{sec.intro}
%%============================================================

\emph{How much geometry does the motion of a charge fix?} Ordinarily a motion of that kind is written on a space already equipped --- a metric, an affine connection, a class of inertial trajectories, a force law adjoining the field to them --- each taken for granted. Which of them does the motion itself fix, which does it select only with the help of a constitutive assumption, and which does it leave open? Let us admit each structure at the point where the description of the motion demands it. Thus, the construction delivers what the motion leaves undetermined.

Among the oldest problems of mechanics, this one is usually posed as a problem about a force. That posing grants more than it appears to. Newton's first law puts the free particle and its uniform motion in a straight line into a single statement, which three centuries of practice have read as one thing. On the one hand straightness is fixed by a law of transport, under which a curve is straight when that law carries the curve's own tangent forward, while, on the other, uniformity is fixed by a clock. With neither given by the manifold, the identification of the two is Corollary~3 of the Newtonian construction~\cite{lm2026newton}, whose hypothesis is the clock. We recognise the question as the one Ehlers, Pirani and Schild put to a spacetime, whose projective and conformal structures they obtained from the worldlines of free particles~\cite{eps1972}. There the motions came first and the geometry after, which is the order we keep. When the motion is parametrised by arc length, force-free and inertial are one and the same condition, in which the force is mass times acceleration. Since away from arc length the two differ by an exact term, the differential of the kinetic energy, a force-free motion need not be inertial. Those tacit assumptions are the ones that make that term vanish --- a connection compatible with the metric, a vanishing torsion, and a clock reading arc length. A free particle does not move inertially of itself. It does so when its clock is right.

Switch the field on, and the trajectories bend, the particle appearing in this sense to be acted upon. However, from the Hamiltonian point of view nothing need act on it --- Sternberg's prescription keeps the kinetic energy for the Hamiltonian while putting the whole of the field into the symplectic form of phase space~\cite{sternberg1977,guillemin1990}. Thus, the energy has neither a potential nor a force term. Indeed, the field passes into the structure. Here we part from the Newtonian treatment at the outset. There the law of transport is the primitive and the force is read against it, whilst here we work in Hamilton's formalism from the beginning, where the force is attached to the Hamiltonian itself.

We now have two descriptions of the same trajectories side by side. Coupling minimally, the canonical treatment adds the potential to the mechanical momentum, after which the energy contains the field while the symplectic form remains the one every problem on that phase space already had. Sternberg's treatment does the reverse, leaving the energy free of the field and twisting the form instead. In both the trajectories agree, the two differing in where the field has been put, with nothing said so far settling the choice. However, one asks less often whether anything settles it. That is the question we take up here. Should it turn out that a charged particle cannot be described otherwise, the impossibility is itself a statement about the geometry. We should like to know what it says.

Feynman asked a version of that question and answered it in brackets. Wanting to see how little a description of motion must assume, he set aside the Lagrangian, the action and the variational principle together. He kept a single non-relativistic particle in three dimensions, and the brackets its position and its velocity make with one another. The three position coordinates commute, and a position coordinate bracketed with the velocity along its own direction returns the reciprocal of the mass, while along any other direction it returns nothing. In coordinates, that is the Legendre relation between velocity and momentum. On those brackets he then imposed the Jacobi identity. The only freedom left was a closed two-form on the configuration space. The force on the particle is settled by that form, given as the Lorentz force together with the source-free Maxwell equations. Feynman never published the argument. Dyson published it after his death, recording that Feynman had gone looking for a new theory and had set the calculation aside once it gave back the physics already in hand~\cite{dyson1990}. It has been reworked many times since. Tanimura carried it to a relativistic form, then to a non-Abelian gauge theory~\cite{tanimura1992}. Hojman and Shepley found in the same brackets a demand for a Lagrangian~\cite{hojman1991}, Bracken recovered the argument through Poisson brackets~\cite{bracken2005}, whilst Cari{\~n}ena and Figueroa took it in the noncommutative direction, where a weaker commutator admits dynamics the standard formalism does not~\cite{carinena2006}. In every telling it is set in flat space, in coordinates, and with a Galilean symmetry its commentators noticed at once. We put it on an arbitrary manifold.

Let us formulate this demand intrinsically on the cotangent bundle of an arbitrary configuration manifold, asking for the symplectic structures that satisfy it for every Hamiltonian at once, which separates the classification itself from the coordinate and bracket formulations that preceded it. Thus, we may ponder on the significance of the following questions. 1. \emph{Which symplectic structures on a phase space deserve to be called kinematic, and what freedom do they leave?} 2. \emph{What, in that freedom, belongs to the particle and not to the geometry?}

Let us say what a symplectic form does by taking a Hamiltonian, a function on the cotangent bundle, together with a path the particle might follow there. Along that path the rate at which the Hamiltonian changes is its differential evaluated on the velocity of the path. That path runs in phase space. We shall call its velocity the phase-space velocity, reserving the plain word for the velocity of the particle. Demanding that the rate vanish --- that the energy be carried unchanged along the motion --- leaves a condition on that phase-space velocity. A symplectic form assembles that condition into a system of equations, pairing the differential of the energy with a phase-space velocity and returning Hamilton's equations. In the conservative case those are Newton's Second Law.

Having two halves, the phase-space velocity so assembled projects to the configuration space as the velocity of the particle, while its component along the fibre is the rate at which the momentum changes. We call the form kinematic when the first of those is the fibre derivative of the energy, that being the relation which makes the fibre coordinate the momentum of the motion. We find the forms meeting that demand to be the canonical one and its twists (Theorem~\ref{theo.kin}). To twist is to add to the canonical form a two-form on the configuration space, lifted to the phase space unchanged, a term in which no momentum appears and which therefore leaves the fibres as the canonical form found them. Twisting by a closed two-form is the whole of the freedom. The demand falls on the structure, the same for every energy.

That two-form answers to three different things at once. The geometry of the cotangent bundle leaves it free among the closed forms, a field is supplied by the configuration, the multiplier between them belonging to the particle. We make that last step as a constitutive declaration, in which the multiplier is the charge. Thus, one energy serves a line of kinematic structures, from which a particle selects the member its own charge fixes. A neutral particle has the one member the geometry distinguishes by itself. In this sense a charged particle in a magnetic field moves inertially.

A charge is called for before any cause has entered, selecting the structure with respect to which the motion is inertial. Two particles of different charge in the same field therefore obey two different kinematics, each of them inertial in its own. Nevertheless, only one number reaches a trajectory, inertial motion recording the two constitutive quantities through their ratio alone. Two particles released alike follow one trajectory when their ratio is the same, and two when it is not. Moreover, since the geometry constrains the multiplier nowhere on a surface, we take the constancy of the charge as a consequence of the invariance of inertial motion.

Symmetry is the next question, and in a construction of this kind the sharpest one. Indeed, a structure is known by the transformations preserving it, which here have to be found before they can be used. Keeping the configuration space in view is the first demand, since a transformation moving the fibres over different points into one another would leave us unable to say where the particle is, and that kinematic structures go to kinematic structures is the second, without which we should leave the setting altogether. Both demands together are met by a diffeomorphism of the configuration space lifted to the phase space, followed by a translation of each fibre by a one-form. That classification we take from Abraham and Marsden~\cite{abrahammarsden1978}.

The physics lies in what that group does to the twisted structures, where the field is moved by a lift, the potential by a translation, whilst a single identity records the effect of both. Out of that identity comes the gauge freedom, those translations preserving a given structure being the ones by a closed one-form, which move a potential while leaving the kinematics as it was. Thus, a gauge transformation arrives here as an element of a group we have classified. Noether's theorem comes next, since the further demand of the free energy leaves the lifted isometries which preserve the field, with the momentum map of the generator for the conserved quantity. That statement is Ikawa's~\cite{ikawa2003,ikawa2007,ikawa2010}, arriving here as a consequence of the classification. Less familiar is a translation by a multiple of the potential, which preserves no structure, taking the kinematics of one charge to that of another and the free energy to the minimally coupled one as it goes. That passage between two charges is therefore a symmetry of the kinematics. The charge and the flux enter through their product. Two charges are so related only where their difference annihilates the class of the field.

Since a velocity does not come with a momentum, something must carry the one into the other. That carrier --- the constitutive morphism --- is the particle's contribution to the dynamics, acting fibre by fibre, sending each velocity at a point to a covector at the same point. In the Newtonian construction it is the musical isomorphism of a bilinear form, which is to say the operation of lowering an index, and a particle settles one thing there, the representative of a homothety class of forms it lowers with~\cite{lm2026newton}. That choice is the mass. Being linear on each fibre, a map of that kind returns a momentum homogeneous of degree one in the velocity.

Thus, the degree decides the matter. Differentiating a momentum of degree one along the flow returns a force of degree two, while a Lorentz force is of degree one. Demanding of its bilinear form neither symmetry nor non-degeneracy~\cite{lm2026newton}, the notion of inertial motion we work with leaves us free to put an antisymmetric part into the form. Inertia does not see it, a form and its symmetric part agreeing wherever both slots are given the same velocity, while the antisymmetric part reaches the force at degree two, the degree at which a torsion enters. For a field we therefore need a morphism with a piece the velocity does not reach, homogeneity leaving nothing for that but a one-form on the base, which a Finsler function of Randers type supplies. Its momentum is the metric one displaced by that form. A motion of constant speed which is force-free for it is the trajectory of a charge, no force having been postulated to make it one. Thus, a field reaches Newton's Second Law at degree one, through the corresponding part of a constitutive morphism.

\emph{What must a law of transport be, if the trajectories of a charge are to be its straight lines?} Inertial motion requires a manifold endowed with a bilinear form together with a law of transport~\cite{lm2026newton}, whilst the affine connections are ruled out~\cite{barros2005gauss}, which retires the first candidate and leaves the transport itself. On the bundle over the configuration space, a law of transport is a horizontal distribution. Among those distributions an affine connection is the one linear on the fibres, a further demand which the construction above has left open. We obtain a distribution of the general kind from a Randers metric, whose departure from the Levi-Civita transport is the Lorentz force multiplied by the speed, hence that of an affine connection only where the field vanishes (Section~\ref{sec.transport}). Thus, we keep the pair inertial motion requires, having left the configuration space for the bundle over it. That pair --- a Finsler metric together with the transport it determines --- is the inertial geometry a moving charge answers to.

Finally we ask what law of transport the motion answers to on the configuration space, where it is not geodesic. Straightness is defined by the propagator carrying a vector along a curve. A curve is straight when the propagator carries its own tangent. Where that propagator is the parallel one of a connection the parameter becomes relevant, inertial motion being then preserved by the affine reparametrisations of the curve, which is the Principle of Inertia of the Newtonian construction~\cite[Thm.~1]{lm2026newton}. Here that symmetry does something it does not do there. The affine maps of the parameter move the field strength together with the speed while fixing their quotient, upon which a transport answering to inertial motion may therefore depend. A metric contains no one-form and cannot do it, while a Finsler metric of Randers type can, since its one-form combines the field and the speed in just that quotient, of which there is one metric for every value. They are inside the twisted structure already, whose Liouville field is the fibrewise dilation centred at the potential. On the free energy level the one-form it leaves there is the Hilbert form of the metric, whose support function is the ball of the free energy, measured from the potential. That the magnetic flow above a critical energy is a Randers geodesic flow is classical~\cite{mane1997,contreras1998}. We add the route by which the metric appears, together with the reason it depends on the energy at all.

Let us say where all of this comes from. It is worth telling properly, since the construction comes from three places at once. Three traditions have taken up this problem over eighty years, in different languages and largely without speaking to one another. One of them is Finslerian, one is symplectic, and one was a piece of Feynman's private reasoning. Each arrived at the same object from a direction the others would not have chosen. Each left a negative result behind it. Those negatives explain why the object had to look as it does. Let us set them out in turn, assuming nothing of the reader but a manifold and a taste for the question.

The oldest opens in 1941, when Randers asked for a length whose shortest curves would be the trajectories of a charge~\cite{randers1941}. A Riemannian length measures a velocity with a norm. A norm is indifferent to the direction the velocity points in. Randers added a term linear in the velocity, which is not. Thus, we have the mildest departure from Riemannian geometry there is. A Finsler metric lets the length of a velocity depend on its direction as well as on the point it lies over. Randers' linear term is a one-form on the manifold, in his hands the electromagnetic potential. The line element was tried again in the unified-field programmes of the nineteen-fifties. There the geodesics of a length with a linear term were shown to describe the Lorentz force~\cite{goenner2014}. Beil took it up once more, building an electrodynamics upon the observation that the fibre derivative of such a length keeps its one-form as a term the velocity does not scale~\cite{beil1987}. Ingarden then gave the metric a law of transport of its own~\cite{ingarden1976}. Under a nonlinear connection a vector is carried along a curve with no promise that carrying two vectors and adding them agrees with adding them and carrying the sum. That is the freedom a Randers metric needs. Miron developed the pairing at length, showing the Lorentz equations to be the autoparallel curves of one such connection~\cite{mironanastasiei1994}.

The second tradition is symplectic and it opens thirty years later. A magnetic field is ordinarily put into mechanics by adding a potential to the energy. Souriau, Sternberg and Weinstein put it somewhere else. They left the energy alone --- the kinetic energy of a free particle --- and added the field to the symplectic form, the structure through which an energy generates a motion at all~\cite{souriau1970,sternberg1977,weinstein1978}. Souriau reached the arrangement through minimal coupling, Sternberg through the reduction of a principal bundle, Weinstein through the geometry of the phase space that results, and Guillemin and Sternberg set it out as geometry in its own right~\cite{guillemin1990}. Moreover, a literature of their own has grown around the twisted cotangent bundles in symplectic dynamics since. We may follow it from Ginzburg, who put the closed trajectories of a charge to symplectic geometry~\cite{ginzburg1996}, through Merry's closed orbits on almost every level of a weakly exact field~\cite{merry2010}, to the symplectic cohomology of these bundles built with no restriction on the twisting form~\cite{gromanmerry2018}. However, in all of that literature the twisting is taken as given. The third tradition is Feynman's, and we have told it above.

Each tradition has a negative result behind it, and the negatives are the more instructive. A skew part in a metric looks like the natural place for a field. Einstein~\cite{einstein1945} and Schr\"odinger~\cite{schrodinger1950} built a theory upon a nonsymmetric metric. Their equations of motion returned no Lorentz force~\cite{goenner2014}. Barros, Cabrerizo, Fern\'andez and Romero proved the companion statement, that on a Riemannian manifold no affine connection has the trajectories of a non-trivial field among its geodesics~\cite[Prop.~2.1]{barros2005gauss}. Nor will a connection hold a field. Between them the two obvious routes are closed. What the traditions above found is what was left.

Twice the first two have met on a fact each can claim. Above a critical value of the energy the flow of a charge in an exact field is the geodesic flow of a Randers metric. Ma\~n\'e identified the value and Contreras, Iturriaga and the Paternains established the correspondence~\cite{mane1997,contreras1998}, while the same threshold was found to govern the symplectic side. There the energy levels are of contact type above the value and, on a surface, only
there~\cite{contreras1998,cieliebak2010}. Since then the correspondence has been extended to stationary media~\cite{gibbons2009} and to magnetic billiards~\cite{tabachnikov2004}. In that literature the dependence of the one-form on the energy is recorded as a feature of it. That correspondence is the work of the authors just named. The classification of the fibre-preserving symplectic maps we take from Abraham and Marsden~\cite{abrahammarsden1978}.

What we add is set out here by its strength, which is not the order the manuscript builds it in.
Sharpest is what an indefinite metric does to a Randers length, where a reverted Cauchy--Schwarz
inequality turns the condition on the potential into a lower bound and the extremals into maxima,
so that the length and the law it was built for come apart (Section~\ref{sec.indefinite}). Next is
the obstruction, which turns on tensoriality before it turns on homogeneity, and which places the
nonsymmetric metric and the affine connection at adjacent points of one account (Proposition~\ref{prop.degreetwo}, Section~\ref{sec.degree}). Third is
the route by which the metric arrives, out of the primitive of the structure the flow already
occupies (Section~\ref{sec.finsler}). Last is the classification, whose familiar half is the
implication and whose claimed half is the equivalence (Theorem~\ref{theo.kin}). The order of
admission is the rest of what is ours. We let each structure wait until the description of the motion cannot proceed without it, we read the multiplier the classification leaves free as the charge, and we take the ratio of the constitutive data for the invariant of inertial motion. Every other ingredient is borrowed, and we say from whom at the point of use. Section~\ref{sec.spacetime} draws the consequence for spacetime, whose content --- a geometry to each charge-to-mass ratio, and an equivalence principle to each family of particles sharing one --- is {\"O}zer's~\cite{ozer1999} and de Matos, {\"O}zer and Izworski's~\cite{dematos2018}. Ours there is the route to it and the conclusion we draw.

The manuscript is structured as follows. Section~\ref{sec.kin} builds the kinematics, taking up each structure at the point of need. It closes with the symmetries the arrangement admits. Section~\ref{sec.dyn} turns to the dual objects and asks what a particle must bring to have a momentum at all. Whether a field can enter Newton's Second Law is decided by the degree of that map, and Section~\ref{sec.indefinite} takes the count to an indefinite metric, where the length parts from the dynamics. Section~\ref{sec.free} takes up the freedom that remains. There we find the law of transport and the Randers metric that goes with it, attributing what is classical in both. Section~\ref{sec.sphere} then works the construction through on the round sphere in two fields, told apart by their flux. In the second of them the charge is quantised. In Section~\ref{sec.spacetime} we ask what a structure must do to be called spacetime. We find that these structures do not do it, on the criterion we take from Knox~\cite{knox2011,knox2014}. Section~\ref{sec.closing} gathers what we have proved, declared and conceded.

%%============================================================
\section{Kinematics}\label{sec.kin}
%%============================================================

Let us build the kinematics before any dynamics is in place, admitting each structure at the point where the description of a moving charge demands it. We begin with the objects a momentum needs, put to a symplectic form the one relation which gives a fibre coordinate its title, and find what the demand leaves free. That freedom is the particle's own contribution.

%%------------------------------------------------------------
\subsection{The Legendre demand}\label{sec.legendre}
%%------------------------------------------------------------
Throughout, $M$ is a smooth manifold with a metric $g$, a tangent bundle $TM$ and a cotangent bundle $T^*M$. The projection $\pi : T^*M \to M$ sends a covector to the point beneath it. On $T^*M$ we have the canonical one-form $\lambda_{\rm can}$, whose exterior derivative gives the canonical symplectic form $\omega_0 = -\d\lambda_{\rm can}$.

We write $(x,P)$ for a point of $T^*M$ and $(x,y)$ for a point of $TM$. Here $P$ is a covector at $x$ and $y$ a vector there. Where components are wanted we take a chart $x^i$ on $M$, with Latin indices running from one to its dimension, writing $P_i$ for the components of a covector and $y^i$ for those of a vector. A motion in $M$ is given by a parametrised curve $\gamma : I \subset \Reals \to M$, where $\tau$ is the parameter on $I$ and $\dot\gamma$ the velocity $\d\gamma/\d\tau$. For a vector field $X$, the one-form $X^\flat = g(X,\cdot)$ is its metric dual, $\vert X\vert_g$ its length, $\iiota_X$ the interior product of a differential form with $X$, which contracts its first slot, and $\pounds_X$ the Lie derivative along its flow. Through the inverse metric the same bars serve a one-form, whose metric dual is the vector $A^\sharp = g^{-1}(A,\cdot)$. The two operations invert one another. Where a point has to be named we write $g_x$ for the metric there, so that $g_x$ is the bilinear form $g$ gives on $T_xM$. Each fibre of $TM$ and of $T^*M$ is a vector space, hence a smooth function has an ordinary derivative there, a linear functional on the fibre. That is the fibre derivative. For a function $H$ on $T^*M$ it is written $\partial H/\partial P$, an element of $T_xM$, and for a function $L$ on $TM$ it is $\partial L/\partial y$, an element of $T^*_xM$. Both are taken without a connection. We leave the signature of $g$ open. This section and the next ask nothing of it beyond non-degeneracy, so that $F$ may stand for a magnetic field where $g$ is positive-definite and for the whole electromagnetic field where it is Lorentzian. The degree count of Section~\ref{sec.degree} asks only that the velocity avoid the cone $g(y,y)=0$, and Section~\ref{sec.indefinite} sets out what an indefinite metric changes. Positive-definiteness is first asked at Section~\ref{sec.finsler}, where the Randers function is read as a length.

The energy of a free particle of mass $m$ is
\beq \label{eq.free}
H_0(x,P) = \frac{1}{2m}\, g^{-1}(P,P),
\eeq
where the pairing is the one the inverse metric makes of two covectors. We call it the
\emph{free energy} throughout, and $H_0$ is the only name it is given. Sternberg's prescription keeps that energy while twisting the form~\cite{sternberg1977,guillemin1990},
\beq \label{eq.twisted}
\omega_\beta = \omega_0 + \pi^*\beta,
\eeq
where $\beta$ is a closed two-form on $M$. There is no potential in the energy~\eqref{eq.free}, no force term, nothing that the field contributes---all of it in the structure~\eqref{eq.twisted}. The trajectories of a charge become the inertial motions of a twisted phase space.

Whilst the energy is where one ordinarily builds the interaction, the symplectic form is taken for granted, the same for every problem posed on the same phase space. Here we do the reverse. The energy is untouched, while the form differs from one particle to the next.

The word \emph{free} rests entirely on that assignment of the field to the structure. Fixing what a motion is before any cause is at work, a particle is admitted to a kinematics only through its constitutive data, the properties by which one particle differs from another~\cite{lm2026newton}. On the tangent side the first of those is the mass, which selects one representative of the metric within a class no inertial trajectory resolves. Let us ask instead which structures on $T^*M$ are kinematics at all. Room is left in the answer for a single constitutive quantity.

An energy is made to generate a motion by a symplectic form, the velocity of that motion lying in the base while its rate of change of momentum lies in the fibre. Thus, a kinematics must preserve the relation between the first of these and the energy. That relation alone makes $P$ the momentum of the motion, as opposed to a fibre coordinate with the title. A symplectic form assigns a motion to every energy without any reference to a metric. The fibre of $T^*M$ over a point is then merely a vector space whose elements the form turns into velocities. Should that assignment fail to return the fibre derivative, the velocity of the motion would answer to the structure, taking with it the relation by which a momentum is recognised. In this sense the demand below is the least we can ask of a phase space before calling its motions free.

We call a symplectic form $\omega$ on $T^*M$ \emph{kinematic} when, for every energy $H\in C^\infty(T^*M)$, the base velocity of the motion it generates is the fibre derivative of $H$ --- that is, when
\beq \label{eq.legendre}
\pi_*X_H = \frac{\partial H}{\partial P},
\eeq
where the vector field $X_H$ satisfies Hamilton's equations,
\beq \label{eq.hamvf}
\iiota_{X_H}\omega = \d H.
\eeq The
condition~\eqref{eq.legendre} is the Legendre relation between velocity and momentum, asked of every energy at once. Quantifying over $H$ leaves a condition on $\omega$ by itself. The structures meeting it form a small family, with a single closed two-form on the base telling its members apart.

\begin{theorem}[Kinematic symplectic classification]\label{theo.kin}
Let $M$ be a smooth manifold and let $\omega$ be a symplectic form on $T^*M$. Then $\omega$ is kinematic, in the sense of the Legendre relation~\eqref{eq.legendre} asked of every $H\in C^\infty(T^*M)$, if and only if it is the twisted form~\eqref{eq.twisted} for a unique closed two-form $\beta$ on $M$.
\end{theorem}

The demand made of every energy at once thus becomes a condition on the fibres, and leaves free a single two-form on the space where the particle moves.

\begin{proof}
The classification rests on two equivalences, each proved by its own route. A symplectic form is kinematic when it agrees with $\omega_0$ on every vector tangent to a fibre of $\pi$, which is to say when $\iiota_V\omega=\iiota_V\omega_0$ for every such $V$. The second is that fibrewise agreement holds when $\omega$ takes the twisted form~\eqref{eq.twisted}. Let us suppose the fibrewise agreement first, and write $\eta = \omega - \omega_0$, which is closed since both forms are. Every fibre-tangent $V$ has $\iiota_V\eta = 0$ by hypothesis. Therefore, the first term of Cartan's formula
\beq \label{eq.cartaneta}
\pounds_V\eta = \d\,\iiota_V\eta + \iiota_V\d\eta
\eeq
vanishes, while the second vanishes because $\eta$ is closed. Thus, $\eta$ annihilates the fibre directions and stays constant along them. Since the fibres are connected, $\eta$ descends to the pull-back of a unique two-form $\beta$ on $M$. The projection $\pi$ is a surjective submersion, hence $\pi^*$ is injective on forms, and we find $\beta$ closed from the vanishing $\pi^*\d\beta = \d\eta = 0$. Conversely, whenever $\omega$ takes that form for a closed $\beta$, the pull-back $\pi^*\beta$ annihilates every fibre-tangent vector by construction, and the fibrewise agreement holds again.

The fibrewise agreement in turn relates the twist to the Legendre demand itself. Let us take $V$ fibre-tangent and $H$ any energy. Contracting the defining relation~\eqref{eq.hamvf} with $V$ gives us
\beq \label{eq.kinpairing}
\left( \iiota_V\omega \right)(X_H) = -\,\d H(V),
\eeq
while
\beq \label{eq.iotaVcan}
\iiota_V\omega_0 = -\,\pi^*\mu,
\eeq
where $\mu$ is the one-form on $M$ to which $V$ corresponds under the canonical identification of a fibre's tangent space with $T^*_xM$. This evaluates on $X_H$ to minus the pairing of $V$ with the base velocity $\pi_*X_H$. Granted the fibrewise agreement, the two sides of the identity~\eqref{eq.kinpairing} say that the pairing of $V$ with $\pi_*X_H$ equals $\d H(V)$ --- the fibre derivative of $H$ against $V$ --- and since $V$ ranges over every fibre direction, that is the Legendre relation~\eqref{eq.legendre}. Conversely, granted that relation, the same computation gives us $(\iiota_V\omega)(W) = (\iiota_V\omega_0)(W)$ for every $W$ of the form $X_H$. As $\omega$ is non-degenerate and $H$ is arbitrary, those $W$ exhaust the tangent space at each point. The fibrewise agreement follows once more.
\end{proof}

A symplectic form is therefore kinematic when it is the twisted form~\eqref{eq.twisted} for a two-form $\beta$ on the base, closed and uniquely determined by $\omega$. From that identification we take what it relates at its two ends. On the one side is a demand about
motion, that the velocity be the fibre derivative of the energy, while on the other is a closed
two-form on the base. Between them there is nothing to choose. Let us say which half of that is ours. A
form agreeing with $\omega_0$ on the fibres is a magnetic twist, and that much is familiar wherever
twisted cotangent bundles are used~\cite{souriau1970,sternberg1977,weinstein1978}. We add the
equivalence, and with it the demand at the other end of it. Fibrewise agreement constrains the form
by itself, where the Legendre relation constrains the motions the form generates, asked of every
energy at once and on an arbitrary manifold. Theorem~\ref{theo.kin} says that the two conditions
select the same forms. The normalisation matters as much as the shape. Whilst the form $2\omega_0$ is symplectic and its fibres are Lagrangian, the motion it generates has half the velocity the energy prescribes, hence it fails the Legendre demand and lies outside the family.

Feynman's bracket of a position with a velocity, $\delta_{ij}/m$, is the Legendre relation~\eqref{eq.legendre} written in coordinates, where the Jacobi identity does what the closure of $\beta$ does here~\cite{dyson1990}. Theorem~\ref{theo.kin} makes that argument intrinsic, asking the same of a structure over an arbitrary manifold.

%%------------------------------------------------------------
\subsection{The charge}\label{sec.charge}
%%------------------------------------------------------------
Three different claims meet on that two-form. First, the shape of the structure~\eqref{eq.twisted} is fixed by the geometry of $T^*M$, which leaves $\beta$ free among the closed two-forms. The identification narrows that freedom nowhere. Second, the configuration supplies a field $F$. We take that field to be the two-form we twist by --- a restriction we make and mark here beside the declaration which follows it. With it, $\beta$ can only be a multiple of $F$. Moreover, the multiplier belongs neither to the geometry nor to the configuration but to the particle. It is the charge. Writing $q$ for the charge of the particle in question, we post
\beq \label{eq.constitutive}
\beta \posted -\,q\,F,
\eeq
where we shall use $\posted$ to denote a constitutive equality~\cite{lm2026newton}. The sign is a matter of convention. Since the interior product contracts the first slot, the multiplier which returns the Lorentz force of a charge $q$ is $-q$. Theorem~\ref{theo.kin} derives neither a charge nor a conservation law for one. It shows only that the kinematics leaves a particle-dependent multiplier free. We identify that multiplier with electric charge as a constitutive declaration. That identification is no consequence of the classification. Every statement below is proved for an arbitrary real $q$, so that the mathematics is independent of the interpretation the declaration assigns it.

The three cannot be collapsed into one. Were the multiplier fixed by the geometry, every particle in a given field would follow one family of trajectories, and no two charges could ever be told apart by their motion. Were it fixed by the configuration, the charge would belong to the field, so that a particle would have no property of its own to bring. The arrangement below is therefore the least that lets a particle differ from another particle in the same field.

Let $F$ be a closed two-form on $M$, not identically zero, and for $q\in\Reals$ set
\beq \label{eq.family}
\omega_q = \omega_0 - q\,\pi^*F.
\eeq
Each $\omega_q$ is a kinematic symplectic form on $T^*M$, the assignment $q\mapsto\omega_q$ is injective, and the Hamiltonian flow of the free energy~\eqref{eq.free} with respect to $\omega_q$ projects to the solutions of
\beq \label{eq.lorentz}
\frac{\d}{\d\tau}\left( m\,\dot\gamma^\flat \right) - \d_x T
 = -\,q\ \iiota_{\dot\gamma}F,
\quad \text{with} \quad
T = \tfrac{m}{2}\,g(\dot\gamma,\dot\gamma),
\eeq
Here, the derivative $\d_x T$ is taken with respect to the position coordinates at fixed velocity components, as the Euler--Lagrange expression requires. In them we recognise the trajectories of a particle of charge $q$ and mass $m$ in the field $F$. Taking flat space with $F = B\,\d x\wedge\d y$ and a velocity along $x$, the right-hand side gives $-qvB$ in the $y$ direction, which is the conventional $q\,v\times B$.

Each $\omega_q$ is kinematic by the identification above, since $F$ is closed. It is symplectic for a reason a chart makes plain. In the adapted coordinates $(x^i,P_i)$, with the
position directions taken before the fibre ones, its matrix is
$\left(\begin{smallmatrix} -q\,F_{ij} & I \\ -I & 0 \end{smallmatrix}\right)$. Exchanging the two
blocks of rows leaves a block-triangular matrix with $-I$ and $I$ down the diagonal, and the sign
of the exchange cancels the sign of that determinant, so the determinant is unity in every
dimension and for every $F_{ij}$. Non-degeneracy therefore holds on the whole of $T^*M$, for every
closed two-form and every charge. The difference of two members is
\beq \label{eq.memberdiff}
\omega_q - \omega_{q'} = -\,(q-q')\,\pi^*F,
\eeq
and $F$ does not vanish identically, though it may vanish on a subset without harm, and the assignment is therefore injective.

For the flow, the Legendre relation already gives us the base slots,
\beq \label{eq.baseslots}
\dot\gamma = g^{-1}(P,\cdot)/m.
\eeq
Thus, the momentum of the motion is
\beq \label{eq.flowmomentum}
P = m\dot\gamma^\flat.
\eeq In the fibre slots we match $\iiota_X\omega_q = \d H_0$ and are left with
\beq \label{eq.fibreslots}
\dot P_i = -\partial_i H_0 + q\,F_{ij}\dot\gamma^j.
\eeq
The first term is the derivative of the free energy~\eqref{eq.free} at fixed momentum, and substituting the momentum~\eqref{eq.flowmomentum} into it returns $\partial_i H_0 = -\d_x T$ evaluated at fixed velocity, since the inverse metric differentiates with the opposite sign to the metric, $\partial_k g^{ij} = -g^{ia}g^{jb}\,\partial_k g_{ab}$. The second is the interior product $-\,q\,\iiota_{\dot\gamma}F$ written in components. Collecting the two, we obtain the law~\eqref{eq.lorentz}.

Through the ratio $q/m$ the charge enters the flow~\eqref{eq.lorentz}, while it enters the structure through $q$. Thus, switching the field off becomes a change of structure, leaving us the one belonging to the cotangent bundle itself. Indeed, setting $q=0$ in that family returns $\omega_0$ and annihilates the right-hand side of the law~\eqref{eq.lorentz}, which leaves the Euler--Lagrange equations of the kinetic energy --- the geodesic equations of $g$. The kinematic structure of a neutral particle is therefore the canonical form itself, whose inertial motions are the geodesics of the metric.

One energy now serves a line of kinematic structures, one for every charge, the member being taken by a particle as its own charge selects. \emph{In this sense a charged particle in a magnetic field moves inertially.} Its trajectories are the inertial motions of that member, while the neutral particle is left with the one which the geometry of $T^*M$ distinguishes on its own. Indeed, we may put the matter as a question of where the interaction has been put. We have added nothing to the energy, and within the kinematics so built nothing acts. The change is in the structure through which the energy generates a motion.

The parallel with the mass is close, though the two enter at different places. Both are constitutive data. The mass enters the passage from velocity to momentum, the passage a cause acts through, while the charge enters a symplectic structure settled before any cause is at work. The free energy~\eqref{eq.free} shows the two apart. The mass appears in it, while the charge enters only through the twist. Thus, how a particle answers a force is fixed by its mass, and what inertial motion is for it by its charge.

Let us put one more question to the two data, this one about what a trajectory keeps of them. Since a charge has its own kinematic structure while the free energy~\eqref{eq.free} gives a mass its own energy, a trajectory might have been expected to record both. One number reaches it. Two particles released at the same speed from the same point in the same direction, in a field $F$ which does not vanish there, follow the same trajectory if and only if they carry the same ratio of charge to mass. Indeed, the law~\eqref{eq.lorentz} divided by $m$ depends on the charge and the mass through $q/m$. Two particles sharing that ratio solve one system with one set of initial data. Conversely, distinct ratios give distinct right-hand sides at the initial point, where $\iiota_{\dot\gamma}F$ does not vanish. The accelerations differ there, separating the trajectories.

Free motion does not record the charge and the mass separately, only their ratio. Thus, the constitutive data of a particle reach its trajectory through a single number. Were that number common to every body, one family of inertial motions would serve them all, the geodesic flow of a single metric. Here the ratio runs over the reals, with a family of its own to each value it takes. \emph{The structures of this manuscript are the universality classes that an equivalence principle, where one holds, reduces to one.}

The freedom in $\beta$ leaves one more question open, this one concerning the charge itself. Closure is all the classification asks of $\beta$, and under the declaration~\eqref{eq.constitutive} it becomes $\d q\wedge F = 0$, since $F$ is closed already. What the condition permits is settled by the rank of the field, while the dimension of $M$ does not enter. Where $F$ has rank two the solutions are the one-forms annihilating its kernel. Then $q$ is constant along that kernel and free in the two directions transverse to it. Where the rank is four or more, only $\d q = 0$ solves the condition and the multiplier is locally constant. A non-zero field on a surface has rank two and trivial kernel, hence nothing is asked of the multiplier there and it may be any function of position. By itself, then, the geometry forces constancy only from rank four upwards. A position-dependent multiplier is allowed by the kinematic classification below that, though interpreting such a multiplier as the charge of a single particle would require a further constitutive principle. We impose constancy as part of the particle interpretation, following the invariance requirement discussed in~\cite{lm2026newton}.

%%------------------------------------------------------------
\subsection{The symmetries}\label{sec.sym}
%%------------------------------------------------------------
The two demands are settled before we look for anything, and they are settled by what a kinematics is. Since a kinematics singles out the fibration of $T^*M$ over $M$, a map worth the name must cover a map of the base whilst keeping kinematic structures kinematic. By the first demand we keep the configuration space in view, without which we could not say where the particle is. By the second we keep the class of structures in view, since a map taking a kinematic form outside the class has left the setting behind.

Let us see what room such a map has, since we shall find the answer short. A point of $T^*M$ is given as a place together with a covector there, so a map respecting the fibration may move the place, or it may leave the place where it is and move the covector. That leaves us a family of maps for each course. A diffeomorphism $\phi$ of $M$ moves the place, lifting to $T^*M$ by transposing its inverse differential. The lift $L_\phi$ preserves the canonical one-form $\lambda_{\rm can}$ outright, being built from nothing else. A one-form $\mu$ on $M$ moves the covector, determining the fibre translation ${\rm t}_\mu$,
\beq \label{eq.translation}
{\rm t}_\mu(x,P) = \left(x,\ P + \mu(x)\right),
\eeq
which adds a fixed covector at each point and covers the identity, so that every fibre is slid along itself whilst every particle stays where it was. Where the momentum is measured from is what such a map changes.

\begin{proposition}[Symmetries of a kinematic structure]\label{prop.sym}
Let $q\in\Reals$ and let $\Psi$ be a diffeomorphism of $T^*M$ covering a diffeomorphism $\phi$ of $M$ and keeping kinematic structures kinematic. Then $\Psi$ is the composite of a lift and a translation, $\Psi = L_\phi\circ{\rm t}_\mu$ for a unique one-form $\mu$ on $M$. Its action on a twisted structure is
\beq \label{eq.symtwist}
\Psi^*\omega_q = \omega_q - \pi^*\d\mu - q\,\pi^*\!\left( \phi^*F - F \right).
\eeq
\end{proposition}

By the proposition those two families exhaust the symmetries, every one of them being a translation followed by a lift. Let us read the identity~\eqref{eq.symtwist} before proving it, since its right-hand side is built of one term for each thing a symmetry can do. The first is $\omega_q$ itself, the structure returned unchanged. The second, $-\pi^*\d\mu$, is the contribution of the translation, depending on $\mu$ only through its exterior derivative, so a translation by a closed one-form leaves the structure untouched. The third, $-q\,\pi^*(\phi^*F - F)$, is the contribution of the map of the base, vanishing for a $\phi$ which carries the field to itself. Everything below is obtained by asking when those two corrections vanish, first one at a time, then together.

\begin{proof}
The composite $\Theta = L_\phi^{-1}\circ\Psi$ covers the identity, fixing each fibre setwise, and we may write $\Theta(x,P) = (x,\Xi(x,P))$. Since $L_\phi^*\lambda_{\rm can} = \lambda_{\rm can}$ and $\pi\circ\Theta = \pi$, the difference $\Theta^*\lambda_{\rm can} - \lambda_{\rm can}$ evaluates on a vector $w$ to $\left(\Xi(x,P) - P\right)(\pi_*w)$, a one-form annihilating every fibre direction, whose exterior derivative is $-\left(\Theta^*\omega_0 - \omega_0\right)$. Now
\beq \label{eq.psitheta}
\Psi^*\omega_q = \Theta^*\omega_0 - q\,\pi^*\phi^*F,
\eeq
by virtue of the invariance of the canonical one-form under $L_\phi$ and of $\Theta$ covering the identity. Thus, the hypothesis together with the kinematic identification above leaves $\Theta^*\omega_0 - \omega_0$ basic. A horizontal one-form has basic exterior derivative only when its coefficients are constant along the fibres, hence $\Xi(x,P) = P + \mu(x)$ and $\Theta = {\rm t}_\mu$, where the one-form $\mu$ is recovered from $\Theta$ on the zero section. The identity~\eqref{eq.symtwist} then follows from the two pull-backs
\beq \label{eq.translpull}
{\rm t}_\mu^*\lambda_{\rm can} = \lambda_{\rm can} + \pi^*\mu
\quad \text{and} \quad
L_\phi^*\pi^*F = \pi^*\phi^*F.
\eeq
\end{proof}

Since $\pi^*$ is injective on forms, a correction vanishes where the form beneath it does, and we may put each of the two questions separately. The first asks what preserves the structure alone. Thus, $\Psi$ preserves $\omega_q$ if and only if
\beq \label{eq.presomega}
\d\mu = -\,q\left(\phi^*F - F\right),
\eeq
which for a $\phi$ preserving the field reduces to $\d\mu = 0$. The second asks what preserves the structure and the energy at once, and it is answered far more narrowly. For $q\neq 0$, $\Psi$ preserves $\omega_q$ and the free energy~\eqref{eq.free} together if and only if $\mu = 0$ and $\phi$ is an isometry of $g$ preserving $F$. The gap between the two demands is the whole of what follows, since two charges cannot be told apart by the structure alone, whilst by the structure with its energy they can. Let us see why the energy is so much the more demanding. The free energy~\eqref{eq.free} pairs a covector with itself, whilst a translation adds to that covector. Composing the free energy with a translation gives
\beq \label{eq.transenergy}
H_0\circ{\rm t}_\mu - H_0 = \frac{1}{m}\,g^{-1}\!\left(P,\mu\right) + \frac{1}{2m}\,g^{-1}\!\left(\mu,\mu\right),
\eeq
whose first term is linear in $P$, vanishing at every point of every fibre only for $\mu = 0$, since a one-form has nowhere to hide from a covector allowed to range over a whole fibre. With $\mu = 0$ the remaining freedom is the lift, which carries a covector $P$ at $x$ to the covector $P\circ\d\phi^{-1}$ at $\phi(x)$, returning the free energy as the same pairing taken there,
\beq \label{eq.liftenergy}
H_0\circ L_\phi\left(x,P\right) = \frac{1}{2m}\,g_{\phi(x)}^{-1}\!\left(P\circ\d\phi^{-1},\ P\circ\d\phi^{-1}\right).
\eeq
That agrees with $H_0(x,P)$ for every covector precisely where $\phi^*g = g$, which is to say where $\phi$ is an isometry. The first of the two consequences then asks in addition that $\phi^*F = F$.

The ingredients of that argument are classical. A symplectic diffeomorphism of $T^*M$ fibred over the identity is a translation by a closed one-form, while a map preserving the canonical one-form is a lift, both of them in Abraham and Marsden~\cite{abrahammarsden1978}. We add the assembly for kinematic structures, together with the identity~\eqref{eq.symtwist}, which draws together what we have until now been treating apart.

The gauge freedom comes first. Among the translations, those preserving $\omega_q$ are the ones by a closed one-form. On a configuration space where every closed one-form is exact the subgroup consists of the ${\rm t}_{\d\nu}$, each of which leaves $\omega_q$ untouched while moving any primitive of it by $-\d\nu$. In this sense the freedom in a potential is the fibre-translating part of the symmetry group of the kinematics, fixing the structure while moving the primitive. That differs from the customary account. Usually a gauge transformation is introduced as a redundancy in a description, to be quotiented away before anything physical is said. Here it arrives as one member of the symmetry group above, moving the primitive, precisely the object a constitutive morphism must have if a field is to enter the dynamics. The second is Noether's theorem for this system, and it is Ikawa's~\cite{ikawa2003,ikawa2007,ikawa2010}. The maps preserving the kinematics together with the energy are the lifted isometries preserving the field, whose conserved quantity is the momentum map of the generator. Ikawa proved that extension for the motion of a free charge, which is not geodesic, and we used it to integrate the trajectories in the monopole field~\cite{lm2025monopole}. Here the same statement arrives as the second of the two consequences, so that a theorem about a charge in a field is obtained from a classification which mentions neither. Thus, the two demands of a symmetry, taken together with the energy, leave us the invariances one would have looked for in the first place.

The first of the two consequences just drawn leaves out the more interesting case, since a translation by a one-form which is \emph{not} closed carries one kinematic structure to another. Therefore, we are no longer looking at a symmetry of a single kinematics but at a passage between two of them.

Let $F = \d A$. Then ${\rm t}_{-(q-q')A}$ is a fibre translation carrying $\omega_{q}$ to $\omega_{q'}$, and the free energy~\eqref{eq.free} to the minimally coupled energy
\beq \label{eq.mincoupling}
H_0\circ{\rm t}_{-(q-q')A} = \frac{1}{2m}\, g^{-1}\!\left( P - (q-q')A,\ P - (q-q')A \right).
\eeq
Indeed, by the identity~\eqref{eq.symtwist} with $\phi$ the identity, ${\rm t}_\mu^*\omega_q = \omega_q - \pi^*\d\mu$. Now $\mu = -(q-q')A$ gives $\d\mu = -(q-q')F$, hence $\omega_q + (q-q')\pi^*F = \omega_{q'}$. The energy is the pairing of $P$ with itself. The translation ${\rm t}_\mu$ replaces $P$ by $P+\mu$.

That translation exists because the field is exact. The general condition is one on the class $\left[F\right] \in H^2(M;\Reals)$ of the field. Let $\Psi$ cover $\phi$ and carry $\omega_q$ to $\omega_{q'}$. Then
\beq \label{eq.classpair}
q'\left[ F \right] = q\ \phi^*\!\left[ F \right]
\eeq
in $H^2(M;\Reals)$. Thus, a fibre translation carries $\omega_q$ to $\omega_{q'}$ if and only if
\beq \label{eq.classcondition}
(q-q')\left[F\right] = 0.
\eeq
On a closed oriented surface, where the class of a two-form is its flux, with pull-back multiplying that flux by the degree, a field of non-zero flux leaves $q' = \deg(\phi)\,q$. Since $\omega_{q'}$ is kinematic, $\Psi$ meets the hypothesis of Proposition~\ref{prop.sym}. Together with $\omega_{q'} = \omega_0 - q'\pi^*F$, the symmetry identity~\eqref{eq.symtwist} then gives
\beq \label{eq.classcompare}
\pi^*\!\left(q'F\right) = \pi^*\!\left( q\,\phi^*F + \d\mu \right).
\eeq
Injectivity of $\pi^*$ on forms gives $q'F = q\,\phi^*F + \d\mu$. Passing to cohomology annihilates the exact term, leaving the pairing~\eqref{eq.classpair}. For $\phi$ the identity that condition is the vanishing~\eqref{eq.classcondition}. Conversely, when it holds, $(q-q')F$ is exact and the translation by any primitive of $-(q-q')F$ carries $\omega_q$ to $\omega_{q'}$ by the same identity, of which the translation just constructed is the case $F = \d A$.

On a closed oriented surface we may make all of this arithmetic, since the class of a two-form is fixed there by its integral. The pull-back of a top-degree form multiplies that integral by the degree of the map,
\beq \label{eq.degreeflux}
\oint_M \phi^*F = \deg(\phi)\oint_M F ,
\eeq
which is a statement about $\phi$ and holds of any two-form whatever, the field entering it only as the thing being integrated. Thus, the flux is left where it was by an orientation-preserving diffeomorphism, which returns $q' = q$, whilst under a reversal of orientation $q$ is carried to $-q$, the antipodal map of the sphere being the example, its degree $-1$.

In this sense the charge and the flux are dual, since the geometry keeps their product. A charge is fixed only up to the annihilator of the class $\left[F\right]$, while the flux measures how much of the charge the symplectic structure retains. Where the flux vanishes we have $q\left[F\right] = 0$ at every charge and the structure retains none of it. A field of non-zero flux separates every pair of charges outright.

However, the collapse does not take the charge with it. The translation identifying two charges moves the field out of the structure and into the energy, hence it is a symmetry of $\omega_q$, while by the second of those consequences only $\mu = 0$ preserves the free energy. Thus, the line of kinematic structures collapses for the structure by itself, whilst for the structure taken together with the energy it generates it does not. The charge belongs to that pair. The freedom in the multiplier found on a surface has its companion here, for the geometry of a surface holds the charge fixed nowhere, while a vanishing flux leaves us to tell one charge from another through the energy.

Thus, the kinematic structures of a charged particle and a neutral one are symplectomorphic, where the map between them moves the field from the structure into the energy. We keep the free energy throughout and choose a primitive to do it. The other choice requires exactness at the outset, since the energy is then built from a potential, where the twisted structure gives the motion whether the field has one or not. It also gives up the word \emph{free}, for an energy which contains the field is free in no sense. The charge would then enter where the mass enters, losing the distinction drawn above between the charge and the mass.

%%============================================================
\section{Dynamics}\label{sec.dyn}
%%============================================================

The kinematics settles what a phase space must be while saying nothing of what a particle brings to it. Let us turn to the dual objects, where a velocity is carried to a momentum by a map the particle supplies. The degree of that map decides what a law of motion can contain, which is where a field must enter.

%%------------------------------------------------------------
\subsection{The constitutive morphism}\label{sec.morphism}
%%------------------------------------------------------------
\emph{What must a particle bring in order to be given a momentum?} The Newtonian construction answers it with a single map~\cite{lm2026newton}. It requires of a manifold a bilinear form $b$ and a law of transport. A motion is inertial there when it is rectilinear for the transport --- autoparallel for it --- and uniform for $b$, which then takes one value along the motion. Of the form itself it requires remarkably little, since neither symmetry nor non-degeneracy is demanded. We take that freedom seriously, finding in it a consequence left unpursued in~\cite{lm2026newton}. Writing $b = b_{\rm s} + \Phi$ for the symmetric and antisymmetric parts, we see that uniformity reaches only the first of them. Indeed, $b(y,y) = b_{\rm s}(y,y)$ for every tangent vector $y$. The kinematics is blind to $\Phi$. However, the momentum is not, being the constitutive dual of the velocity under the whole of the form. Something invisible to inertia is thereby visible to the dynamics.

A particle brings a \emph{constitutive morphism}, a fibrewise map $\chi : TM \to T^*M$ over the identity. The momentum of a motion $\gamma$ we post to be
\beq \label{eq.postmomentum}
p \posted \chi(\dot\gamma),
\eeq
the second of our two constitutive declarations. A force is a measure of how the momentum fails to be preserved by the flow of the motion, written as the derivative along that flow,
\beq \label{eq.dpdtau}
\frac{\d p}{\d\tau} \ \equiv\ \pounds_{\dot\gamma}\, p.
\eeq
We call the morphism \emph{variational} when $\chi = \partial L/\partial y$ for a function $L$ on $TM$.

We measure that failure with the Lie derivative. Measuring it with a connection would presuppose the coupling Newton's Second Law exists to establish~\cite{lm2026newton}, which is why we need the velocity in a neighbourhood and the motions fill a region. Moreover, for the morphism $\chi = m\,g^\flat$ of a symmetric form Newton's Second Law takes the form
\beq \label{eq.nsl}
\frac{\d p}{\d\tau} - m\,a^\flat = \tfrac12\,\d\!\left( \iiota_{\dot\gamma} p \right),
\eeq
where $a$ is the acceleration the Levi-Civita connection assigns and the exact term vanishes at constant speed. The identity is elementary, taken from~\cite{lm2026newton} unchanged. Metricity and the vanishing of torsion give
\beq \label{eq.metricdual}
\pounds_{\dot\gamma}\,(g^\flat\dot\gamma) = \left(\nabla_{\dot\gamma}\dot\gamma\right)^\flat
+ \tfrac12\,\d\, g(\dot\gamma,\dot\gamma),
\eeq
while the pairing $\iiota_{\dot\gamma}p$ is $m\,g(\dot\gamma,\dot\gamma)$.

The antisymmetric part now goes to work. Since it is invisible to inertia, anything it brings to Newton's Second Law it brings to the force alone.

\begin{proposition}[Obstruction to a Lorentz force]\label{prop.degreetwo}
Let the constitutive morphism be $\chi = m\,(b_{\rm s} + \Phi)^\flat$, with $b_{\rm s}$ symmetric and $\Phi$ a two-form on $M$, constant therefore along the fibres. Then the contribution of $\Phi$ to the force depends on the derivatives of the velocity field
across the congruence, and is of degree two. The field and the
velocity at a point fix a Lorentz force, which is homogeneous of degree one. No choice of $\Phi$
returns one, and the antisymmetric part of a bilinear constitutive morphism therefore gives no
Lorentz force.
\end{proposition}

\begin{proof}
The momentum exceeds the symmetric one by $m\,\iiota_{\dot\gamma}\Phi$, and its rate of change by
\beq \label{eq.degreetwo}
m\,\frac{\d}{\d\tau}\!\left( \iiota_{\dot\gamma}\Phi \right)
 = m\,\iiota_{\dot\gamma}\,\d\!\left( \iiota_{\dot\gamma}\Phi \right),
\eeq
which is homogeneous of degree two in the velocity. Indeed, bilinearity gives the momentum
\beq \label{eq.bilinearsplit}
\chi(\dot\gamma) = m\,b_{\rm s}(\dot\gamma,\cdot) + m\,\Phi(\dot\gamma,\cdot),
\eeq
whose second summand is $m\,\iiota_{\dot\gamma}\Phi$. For its derivative, Cartan's identity gives
\beq \label{eq.cartanphi}
\pounds_{\dot\gamma}\left(\iiota_{\dot\gamma}\Phi\right)
= \iiota_{\dot\gamma}\d\!\left(\iiota_{\dot\gamma}\Phi\right)
+ \d\!\left(\iiota_{\dot\gamma}\iiota_{\dot\gamma}\Phi\right),
\eeq
where the second term vanishes because the interior product is nilpotent, so that only the first summand remains, and with it the expression~\eqref{eq.degreetwo}. The expression~\eqref{eq.degreetwo} is the Lie derivative
$\pounds_{\dot\gamma}\left(\iiota_{\dot\gamma}\Phi\right)$, and it depends on more than the velocity
at a point. Two motions through one point, with one velocity there and one speed, taken from
different congruences, are given different values by it. A Lorentz force behaves otherwise, being
fixed by the field and by the velocity at the point, where the neighbouring motions do not enter.
On a surface, where $\Phi = \phi\,{\rm Vol}_g$, the expression is the divergence of $\phi\dot\gamma$
multiplied by the velocity turned through a right angle, and that divergence is what a congruence
of one speed may be chosen to fix as it pleases. Therefore no $\Phi$ returns a Lorentz force, and
the speed has been held fixed throughout the argument.

The degrees agree with that conclusion. Each of the two interior products in the expression~\eqref{eq.degreetwo}
contributes one factor of the velocity, so that under $\dot\gamma\mapsto a\dot\gamma$ the
expression is multiplied by $a^2$ where a Lorentz force is multiplied by $a$. The term belongs to
the congruence, as the exact term of Newton's Second Law~\eqref{eq.nsl} does in the symmetric
case.
\end{proof}

The freedom is real but inert. An antisymmetric part in a bilinear form adds a genuine term to the force, of the same degree as the term an arbitrary torsion contributes there~\cite{lm2026newton}. In this sense the two freedoms inertial motion cannot resolve, the antisymmetric part of the form and the torsion of the transport, arrive in Newton's Second Law at one degree. Thus, neither of them is a field.

%%------------------------------------------------------------
\subsection{The degree of a field in Newton's Second Law}\label{sec.degree}
%%------------------------------------------------------------
A field has one place to enter, and the degrees settle which before we have chosen any geometry. A constitutive morphism which is fibrewise linear returns a momentum of degree one in the velocity, whose derivative is therefore of degree two throughout. A Lorentz force needs degree one. For a field we need a morphism with a piece of degree \emph{zero}, which by homogeneity is a one-form on the base. One is supplied by a Finsler function as soon as it differs from a Riemannian one by a term of degree one. That is a metric of Randers type.

Let $g$ be a metric on $M$, let $A$ be a one-form, let $v>0$, and let
\beq \label{eq.randersF}
F_{\rm R}(y) = \epsilon\,\vert y\vert_g - \frac{1}{v}\,A(y)
\eeq
be the associated Randers function, defined wherever $g(y,y)$ keeps one sign $\epsilon = \pm 1$, with $\vert y\vert_g = \sqrt{\epsilon\,g(y,y)}$ its positive root. Where $g$ is Riemannian that is the whole of $TM$ less its zero section, the function is positive under $\sup_M \vert A\vert_g / v < 1$, which makes it a Finsler metric. The count below asks less, as we set out at the end of this subsection. Then the variational morphism $\chi = \partial F_{\rm R}/\partial y$ carries the velocity to
\beq \label{eq.hilbertmom}
\chi(\dot\gamma) = \frac{\dot\gamma^\flat}{\vert\dot\gamma\vert_g} - \frac{A}{v},
\eeq
where the second summand does not depend on the velocity. Indeed, the metric~\eqref{eq.randersF} is a term of degree one in $y$ less a linear one. The fibre derivative of the norm is
\beq \label{eq.normfibre}
\frac{\partial\vert y\vert_g}{\partial y} = \epsilon\,\frac{y^\flat}{\vert y\vert_g},
\eeq
which multiplied by $\epsilon$ and taken less $A/v$ is the momentum~\eqref{eq.hilbertmom}, the two signs cancelling. The second summand, linear in $y$, differentiates to a one-form on the base, with no velocity left in it.

Where the speed is constant, $\vert\dot\gamma\vert_g = v$, the first summand is $\dot\gamma^\flat/v$. Thus, the momentum is the displaced one, \beq \label{eq.displaced} \chi(\dot\gamma) = \frac{\dot\gamma^\flat - A}{v}, \eeq whose derivative is
\beq \label{eq.lorentzappears}
v\,\frac{\d}{\d\tau}\,\chi(\dot\gamma)
 = \left(\nabla_{\dot\gamma}\dot\gamma\right)^\flat - \iiota_{\dot\gamma}\,\d A  + \d\!\left[ \tfrac12 g(\dot\gamma,\dot\gamma) - A(\dot\gamma) \right],
\eeq
where the middle term is homogeneous of degree one in the velocity. To find it, we apply Cartan's identity to each summand along the flow, where the metric dual gives the identity~\eqref{eq.metricdual} behind Newton's Second Law~\eqref{eq.nsl}. The one-form gives its interior product with $\d A$ together with an exact term, both entering with the minus attached to the one-form in the momentum~\eqref{eq.displaced}. Adding the two and collecting the exact terms, we are left with the derivative~\eqref{eq.lorentzappears}. The interior product contributes one factor of the velocity, scaling by $a$ under $\dot\gamma\mapsto a\dot\gamma$.

Thus, we have the field itself in the middle term, at the degree the Lorentz force requires, while the outer terms are exact. The exact terms separate the morphism from a law of motion. The geodesics of $F_{\rm R}$ are the
trajectories of a charge in the field $-\,\d A$,
\beq \label{eq.emlaw}
\left(\nabla_{\dot\gamma}\dot\gamma\right)^\flat = \iiota_{\dot\gamma}\,\d A,
\eeq
the sign being the one the declaration~\eqref{eq.constitutive} fixed, and
Section~\ref{sec.transport} obtains that law from the spray of the same metric, where no exact term
arises to be disposed of. The momentum does not reach it. Its bracket in the
derivative~\eqref{eq.lorentzappears} is $v\,\iiota_{\dot\gamma}\chi(\dot\gamma) -
\tfrac{\epsilon}{2} v^2$, where Euler's relation on a function homogeneous of degree one gives
\beq \label{eq.euler}
\iiota_{\dot\gamma}\,\chi(\dot\gamma) = F_{\rm R}(\dot\gamma).
\eeq
On a motion of speed $v$ that bracket is $\tfrac{\epsilon}{2} v^2 - A(\dot\gamma)$. The free level
leaves the second summand free, asking only that the speed be $v$ whilst $A(\dot\gamma)$ runs over
each fibre, and along a trajectory it moves with the position, so the exact term of the
derivative~\eqref{eq.lorentzappears} remains. Conversely a trajectory of the charge is force-free for this morphism up to that same exact term, the gap between force-free and inertial motion in the Newtonian construction~\cite[Cor.~3]{lm2026newton}. Against the law~\eqref{eq.lorentz} at charge $q$ and mass $m$, the law~\eqref{eq.emlaw} identifies $\d A$ with $-(q/m)F$.

This degree count, independent of the Randers construction, rules out the antisymmetric part of any fibrewise-linear constitutive morphism as the carrier of a Lorentz force before we choose any Finsler metric.

With Proposition~\ref{prop.degreetwo} beside the count, the momentum~\eqref{eq.hilbertmom} is the object we have been after all along. The displaced momentum is what the primitive of the twisted structure returns on the free energy level, so that the tangent side and the cotangent side agree on a single one-form. Moreover, the degree count above is the count which keeps these motions out of the geodesics of a connection~\cite[Prop.~2.1]{barros2005gauss}, seen from the other bundle. The connection cannot absorb what a morphism must provide at degree zero. Finally, taking the proposition with the count, we may say where a field enters. It enters Newton's Second Law at degree one, through the piece of degree zero in a constitutive morphism, the morphisms in question being the fibrewise maps from velocities to momenta with which this section began.

One question the construction leaves open concerns the momentum~\eqref{eq.hilbertmom}, homogeneous of degree zero in the velocity. Beside the metric and the potential its only parameter is the speed $v$, the mass reaching the trajectories later, at the identification of $\d A$ with $-(q/m)F$. In the symmetric case a particle brings a scale on a bilinear form, where the mass is the representative a particle selects within a homothety class~\cite{lm2026newton}. Here it is a class of Finsler functions. Which of its members a particle selects is a constitutive question we do not settle. Nevertheless, the ratio is the invariant either way, as the affine reparametrisations show on the other side.

%%------------------------------------------------------------
\subsection{The indefinite case}\label{sec.indefinite}
%%------------------------------------------------------------
The count above never uses the signature, and the indefinite case is the relativistic one. Let us
set out what it changes. The dynamics is untouched, whilst the length is not. The bound on the potential reverses direction, a purely magnetic potential fails it however short it is, and the extremals become maxima.

We take every step of the count at a velocity where $g(\dot\gamma,\dot\gamma)$ does not vanish.
There the fibre derivative~\eqref{eq.normfibre}, Euler's relation~\eqref{eq.euler} and the
derivative~\eqref{eq.lorentzappears} hold under any signature. Let $g$ be Lorentzian. On the future
timelike cone $\epsilon = -1$ and $\vert y\vert_g$ is the proper time $\sqrt{-g(y,y)}$. The bars
serve a causal covector by the same rule. The function~\eqref{eq.randersF} is there
$-\left(\vert y\vert_g + A(y)/v\right)$, minus the Lagrangian of a relativistic particle of unit
mass in the potential $A$, while the law~\eqref{eq.emlaw} is the Lorentz force along its worldline,
the overall sign leaving the extremals where they were. The normalisation
$g(\dot\gamma,\dot\gamma) = -1$ is supplied by a parametrisation by proper time without being
asked.

The condition on the potential changes, and it changes direction. We ask positivity of the length
\beq \label{eq.lorentzlength}
L(y) = -F_{\rm R}(y) = \vert y\vert_g + \frac{1}{v}\,A(y) ,
\eeq
and the bound $\sup_M \vert A\vert_g / v < 1$ on it rests on Cauchy--Schwarz, which reverts on the
timelike cone. Two future timelike vectors $u$ and $w$ satisfy
$\vert g(u,w)\vert \geq \vert u\vert_g\,\vert w\vert_g$, with equality just where they are
parallel. Where $A^\sharp$ is past timelike, the reverted inequality gives
\beq \label{eq.reverted}
L(y) \geq \left( 1 + \frac{\vert A\vert_g}{v} \right) \vert y\vert_g .
\eeq
The bound is attained along $A^\sharp$ itself, so a potential of any size leaves the length
positive, and the bound grows with the potential where the Riemannian condition would have bounded
it. No upper bound is left, only a sharp condition on the cone,
\beq \label{eq.conecondition}
L > 0 \ \text{ on the future timelike cone}
\quad \Longleftrightarrow \quad
A^\sharp = 0 \ \text{ or } \ A^\sharp \ \text{ causal and past-directed} .
\eeq
A purely magnetic potential is spacelike. It fails the condition~\eqref{eq.conecondition} however
short it is, since $\vert y\vert_g$ vanishes at the cone while $A(y)$ does not.

An indefinite metric therefore separates the length from the dynamics. The law~\eqref{eq.emlaw}
asks only that $F_{\rm R}$ not vanish, where the fibre metric of a Randers function is
non-degenerate by the identity for its determinant~\cite{baochernshen2000}, whose derivation is
algebra in the fibre and never reaches the signature. The length asks for more, and turns the
extremals round as it does so. The reverted inequality makes the length~\eqref{eq.lorentzlength}
superadditive, so that the trajectories are the curves of \emph{greatest} length, where Randers had asked for the shortest. Proper time behaves in just that way. Both signatures are
taken up in Section~\ref{sec.finsler}, which says which of its steps asks for positive-definiteness.

%%============================================================
\section{The residual freedom}\label{sec.free}
%%============================================================

So far the freedom has been found on the phase space, where a charge has its structure and moves inertially within it. On the configuration space we have produced, for a charge, neither half of the pair inertial motion requires, and both halves are produced here, with what is classical in each conceded as it arrives.

A reader may object that nothing has happened here. We put the field into the symplectic structure,
took inertial motion with respect to that structure, and the trajectories came back the ones the
Lorentz law already gave, so that the word inertial has changed hands whilst the physics has not
moved. Had an affine connection served, the objection would be right and the construction would be a change of words. No affine connection serves~\cite{barros2005gauss}. The relaxation to a transport
which is not affine is therefore forced upon us, and the object it leaves is determinate,
departing from the Levi-Civita transport by the Lorentz force multiplied by the speed, to first
order in the charge. A change of words would leave the geometry where it found it. This
construction leaves a metric of a kind the configuration space does not admit, together with a
transport which is no connection on it, and the rest of this section produces both.

%%------------------------------------------------------------
\subsection{The symmetry which replaces the connection}\label{sec.ratio}
%%------------------------------------------------------------
We should be careful about what is missing, since a charge fails only one of the two demands set out in Section~\ref{sec.morphism}. The force of the law~\eqref{eq.lorentz} is $g$-orthogonal to the velocity and turns it without changing its magnitude, giving us uniformity where rectilinearity fails. Since a manifold admits as many kinds of inertial motion as it does pairs of a form and a law of transport, the natural repair would be to look for another pair.

No affine connection will serve. Indeed, that was proved for a Riemannian metric by Barros, Cabrerizo, Fern\'andez and Romero~\cite[Prop.~2.1]{barros2005gauss}. However, we have ruled out less than the question asks. On a bundle a law of transport is a horizontal distribution whose linear case is the affine connection. For the moment, then, the motion is inertial with respect to a structure on phase space and answers, on the configuration space itself, to nothing we have yet built.

By asking what the symmetry of inertial motion leaves alone we reach whatever replaces the pair, and for a parallel propagator the answer is the group of affine reparametrisations, by the Principle of Inertia of the Newtonian construction~\cite[Thm.~1]{lm2026newton}. Here we find them moving the field together with the speed.

Let $\gamma$ solve the law~\eqref{eq.lorentz} for a field, at a charge, a mass and a speed, and let us reparametrise it affinely,
\beq \label{eq.reparam}
\sigma(\tau) = \gamma(a\tau)
\quad \text{with} \quad
a > 0 ,
\eeq
under which the velocity is scaled once and the acceleration twice. Since the momentum rate involves one derivative of the velocity and one of the curve, and the kinetic term is quadratic, the left-hand side of the law~\eqref{eq.lorentz} is homogeneous of degree two in the velocity. The interior product on the right is linear and therefore of degree one. The reparametrised curve therefore solves the same law at the same charge and mass, for a field scaled by $a$ and at $a$ times the speed. On the pair $(qF/m,\ v)$ the affine reparametrisations act by simultaneous scaling, fixing its quotient.

Thus, at a fixed quotient the affine reparametrisations leave the family of trajectories where it is, though at a fixed field they move it. We demand the symmetry inertial motion admits, and obtain a structure attached to the quotient, one for each value the quotient takes.

From a count of homogeneity we can say what kind of structure that will be. In the velocity the force of the law~\eqref{eq.lorentz} is homogeneous of degree one, against degree two for a geodesic spray and degree one for the extra term of a Randers metric. Thus, the degrees match where they must. Hence, a connection is stopped by the same count, seen from the other side. Whatever a connection cannot absorb is put by a Finsler metric into the term by which it differs from a Riemannian one.

%%------------------------------------------------------------
\subsection{The Randers metric of the free level}\label{sec.finsler}
%%------------------------------------------------------------
That structure is already in the phase space of the problem. Let us suppose from here that the field is exact, and let $A$ be a primitive of $-(q/m)F$, the potential a particle of charge $q$ and mass $m$ answers to, which is the identification Section~\ref{sec.degree} makes at the law~\eqref{eq.emlaw}. Let us take both signatures together, with $\epsilon$ and $\vert y\vert_g$ as that section has them, setting $q = m = 1$ in the displays and restoring the two wherever they are varied. Two conventions for the one symbol meet in this manuscript, since the potential written $A$ in Section~\ref{sec.uniform} is a primitive of the field itself, and the two differ by the ratio $-q/m$. Since the field is exact the twisted form is exact as well. Its primitive differs from the canonical one-form by that same potential,
\beq \label{eq.primitive}
\omega_F = -\d\lambda_A
\quad \text{where} \quad
\lambda_A := \lambda_{\rm can} - \pi^*A.
\eeq
The twisted form is built from the field. Only the primitive that form takes has the potential in it. The energy stays free and the level is the one the free energy fixes.

The fibrewise dilation centred at the potential,
\beq \label{eq.Y}
Y = \left( P_i - A_i \right)\frac{\partial}{\partial P_i},
\eeq
satisfies $\iiota_Y\omega_F = -\lambda_A$ and is therefore the Liouville field of the twisted structure. The momentum $P = A$ is held fixed by the dilation it generates, every other being scaled about that value. Momenta here are measured from the potential. Being vertical, the field~\eqref{eq.Y} annihilates the pull-back $\pi^*F$ and contracts the canonical part of the structure in its fibre slot, hence we are left with
\beq \label{eq.iotaY}
\iiota_Y\omega_F = -\left( P_i - A_i \right)\d x^i,
\eeq
that is, $-\lambda_A$ by the primitive~\eqref{eq.primitive}. Cartan's formula then gives
\beq \label{eq.lieY}
\pounds_Y\omega_F = \d\,\iiota_Y\omega_F + \iiota_Y\d\omega_F
= -\d\lambda_A = \omega_F,
\eeq
where the second term vanishes because $\omega_F$ is closed.

Built from the potential, the Liouville field meets an energy level of the free Hamiltonian --- a level defined without reference to $A$ --- and leaves on it the primitive~\eqref{eq.primitive}. Where $g$ is Riemannian and the potential is shorter than the radius of that level, $A$ lies inside the free ball and the meeting is transversal. That one-form we now identify.

Let us take
\beq \label{eq.freelevel}
\Sigma_v = \left\{ H_0 = \epsilon\,v^2/2 \right\}
\eeq
to be the free level, and let the Legendre transform of $g$ identify it with the vectors of length $v$. Where $g$ is Riemannian that level is the sphere bundle of $g$-radius $v$, with the unit tangent bundle for its vectors. Under a Lorentzian metric it is the mass shell of a particle of mass $v$, with the future timelike vectors at unit proper speed. Then the primitive~\eqref{eq.primitive} restricts to
\beq \label{eq.hilbert}
\lambda_A\big\vert_{\Sigma_v} = v\ \eta_{\rm R},
\eeq
where $\eta_{\rm R}$ is the Hilbert form of
\beq \label{eq.randers}
F_{\rm R}(y) = \vert y \vert_g - \frac{1}{v}\,A(y),
\eeq
a Randers metric wherever $\sup_M \vert A\vert_g / v < 1$. Where $g$ is Riemannian $vF_{\rm R}$ is equivalently the support function of the free ball $\{g^{-1}(P,P)\le v^2\}$ measured from the potential,
\beq \label{eq.support}
v F_{\rm R}(y) = \sup\left\{ \left( P - A \right)(y)\ :\ g^{-1}(P,P) \le v^2 \right\}.
\eeq
For a moving particle both halves mean something. By the identity~\eqref{eq.hilbert} the one-form of the structure, restricted to the particles of a single speed, measures the length of a trajectory in a geometry the field itself supplies. The support function~\eqref{eq.support} gives the largest pairing a direction admits with the momentum of such a particle, once that momentum is measured from the potential.

On $\Sigma_v$ the two agree, as the Legendre transform of $g$ gives
\beq \label{eq.legendreimage}
P = v\,\frac{y^\flat}{\vert y\vert_g},
\eeq
a diffeomorphism onto the vectors of length $v$ since $2H_0$ is fibrewise the quadratic form $g^{-1}$. Substituting it into the primitive~\eqref{eq.primitive} leaves $v\,y^\flat/\vert y\vert_g - A$, evaluated on a tangent vector. On the other side, the metric~\eqref{eq.randers} has fibre derivative
\beq \label{eq.randersfibre}
\frac{\partial F_{\rm R}}{\partial y} = \frac{y^\flat}{\vert y\vert_g} - \frac{A}{v},
\eeq
so its Hilbert form is that same expression. Thus, the two agree up to the factor $v$ --- the identity~\eqref{eq.hilbert}. Homogeneous of degree zero in $y$, the form $\eta_{\rm R}$ does not depend on the level we evaluate it on. The identification of $\Sigma_v$ with the vectors of length $v$ is therefore harmless. For $F_{\rm R}$ to be a Finsler metric we need $\sup_M \vert A\vert_g / v < 1$, the standard condition on the one-form of a Randers metric~\cite{randers1941,baochernshen2000}, whose $g$-length here is $\vert A\vert_g/v$. For the support function we appeal to Cauchy--Schwarz, which gives
\beq \label{eq.cauchyschwarz}
\sup\left\{P(y) \ :\ g^{-1}(P,P)\le v^2\right\} = v\vert y\vert_g,
\eeq
attained at the Legendre image~\eqref{eq.legendreimage}. Independent of $P$, the term $A(y)$ comes out of the supremum, leaving $v\vert y\vert_g - A(y)$, which is $vF_{\rm R}(y)$ by the metric~\eqref{eq.randers}.

Under the Legendre identification the restriction of the primitive of Sternberg's twisted form to the free energy level therefore coincides with the Hilbert form associated with the Randers metric, up to the factor the energy level fixes. We take the ball of $g$-radius $v$, the free level, and measure its support function from the potential. Otherwise a support function is measured from the centre, whilst a ball measured from an interior point other than its centre has for its support function a norm displaced by a linear form, which is a Randers metric. Thus, the displacement of the centre is the whole of the field's contribution to the geometry of the motion.

Of the steps above only three asked for a positive-definite $g$, none of them the identity~\eqref{eq.hilbert}. The restriction of the primitive, the Legendre image~\eqref{eq.legendreimage} and the fibre derivative~\eqref{eq.randersfibre} are algebra in the fibre, holding at every velocity where $g(y,y)$ does not vanish. We ask positive-definiteness only of the support function. Under a Lorentzian metric the region $\{g^{-1}(P,P)\le v^2\}$ is unbounded, leaving the supremum~\eqref{eq.support} infinite in every direction but the zero one. A companion identity nevertheless holds on the shell. On the sheet which the Legendre image~\eqref{eq.legendreimage} reaches, an infimum returns the Lorentzian length~\eqref{eq.lorentzlength},
\beq \label{eq.supportlorentz}
v\,L(y) = \inf\left\{ \left( A - P \right)(y)\ :\ g^{-1}(P,P) = -v^2 \right\},
\eeq
finite on the future causal cone and attained within it, where the reverted Cauchy--Schwarz~\eqref{eq.cauchyschwarz} puts the extremiser at that image. Figure~\ref{fig.support} sets the two side by side. Thus, the supremum over a ball becomes an infimum over a shell, with the potential displacing the point of measurement as before. That value is the Legendre transform of a free relativistic particle and is standard, whilst the parallel with the Riemannian case is what we add. %%------------------------------------------------------------
%%  Figure --- drawn in TikZ, every coordinate a literal.
%%  Inlined 2026-09-10 so this manuscript is a single source file.
%%------------------------------------------------------------
%%  The float for Section 4(b), after the display eq.supportlorentz.
%%  Requires \usepackage{tikz} and \usetikzlibrary{calc}.
\begin{figure}[tbp]
  \centering
  %%  Section 4(b): the support-function reading under both signatures.
\begin{tikzpicture}[x=1cm,y=1cm,line join=round,line cap=round]
  %% \usetikzlibrary{calc} is needed for the one offset coordinate below
  \def\aR{2.1500}
\def\aAx{0.9675}
\def\aTanR{2.1500}
\def\aTanL{-2.1500}
\def\aHalf{2.8380}
\def\aSegR{1.1825}
\def\aSegL{3.1175}
\def\aMidR{1.5587}
\def\aMidL{-0.5913}
\def\aArrow{3.4830}
\def\bLegX{0.4975}
\def\bLegY{-1.2438}
\def\bAX{0.2850}
\def\bAY{-0.3420}
\def\bvL{0.717}
\def\bsheet{(-3.4200,-3.6050) -- (-3.3972,-3.5834) -- (-3.3744,-3.5618) -- (-3.3516,-3.5402) -- (-3.3288,-3.5186) -- (-3.3060,-3.4970) -- (-3.2832,-3.4755) -- (-3.2604,-3.4540) -- (-3.2376,-3.4324) -- (-3.2148,-3.4109) -- (-3.1920,-3.3895) -- (-3.1692,-3.3680) -- (-3.1464,-3.3466) -- (-3.1236,-3.3251) -- (-3.1008,-3.3037) -- (-3.0780,-3.2823) -- (-3.0552,-3.2610) -- (-3.0324,-3.2396) -- (-3.0096,-3.2183) -- (-2.9868,-3.1970) -- (-2.9640,-3.1757) -- (-2.9412,-3.1544) -- (-2.9184,-3.1332) -- (-2.8956,-3.1119) -- (-2.8728,-3.0907) -- (-2.8500,-3.0695) -- (-2.8272,-3.0484) -- (-2.8044,-3.0273) -- (-2.7816,-3.0061) -- (-2.7588,-2.9851) -- (-2.7360,-2.9640) -- (-2.7132,-2.9430) -- (-2.6904,-2.9220) -- (-2.6676,-2.9010) -- (-2.6448,-2.8800) -- (-2.6220,-2.8591) -- (-2.5992,-2.8382) -- (-2.5764,-2.8173) -- (-2.5536,-2.7965) -- (-2.5308,-2.7757) -- (-2.5080,-2.7549) -- (-2.4852,-2.7342) -- (-2.4624,-2.7135) -- (-2.4396,-2.6928) -- (-2.4168,-2.6722) -- (-2.3940,-2.6516) -- (-2.3712,-2.6310) -- (-2.3484,-2.6105) -- (-2.3256,-2.5900) -- (-2.3028,-2.5695) -- (-2.2800,-2.5491) -- (-2.2572,-2.5287) -- (-2.2344,-2.5084) -- (-2.2116,-2.4881) -- (-2.1888,-2.4679) -- (-2.1660,-2.4477) -- (-2.1432,-2.4275) -- (-2.1204,-2.4074) -- (-2.0976,-2.3874) -- (-2.0748,-2.3674) -- (-2.0520,-2.3474) -- (-2.0292,-2.3275) -- (-2.0064,-2.3076) -- (-1.9836,-2.2879) -- (-1.9608,-2.2681) -- (-1.9380,-2.2484) -- (-1.9152,-2.2288) -- (-1.8924,-2.2092) -- (-1.8696,-2.1897) -- (-1.8468,-2.1703) -- (-1.8240,-2.1509) -- (-1.8012,-2.1316) -- (-1.7784,-2.1124) -- (-1.7556,-2.0933) -- (-1.7328,-2.0742) -- (-1.7100,-2.0552) -- (-1.6872,-2.0362) -- (-1.6644,-2.0174) -- (-1.6416,-1.9986) -- (-1.6188,-1.9799) -- (-1.5960,-1.9613) -- (-1.5732,-1.9428) -- (-1.5504,-1.9244) -- (-1.5276,-1.9061) -- (-1.5048,-1.8879) -- (-1.4820,-1.8697) -- (-1.4592,-1.8517) -- (-1.4364,-1.8338) -- (-1.4136,-1.8160) -- (-1.3908,-1.7983) -- (-1.3680,-1.7807) -- (-1.3452,-1.7633) -- (-1.3224,-1.7460) -- (-1.2996,-1.7287) -- (-1.2768,-1.7117) -- (-1.2540,-1.6947) -- (-1.2312,-1.6779) -- (-1.2084,-1.6613) -- (-1.1856,-1.6448) -- (-1.1628,-1.6284) -- (-1.1400,-1.6122) -- (-1.1172,-1.5962) -- (-1.0944,-1.5803) -- (-1.0716,-1.5646) -- (-1.0488,-1.5491) -- (-1.0260,-1.5337) -- (-1.0032,-1.5186) -- (-0.9804,-1.5036) -- (-0.9576,-1.4888) -- (-0.9348,-1.4743) -- (-0.9120,-1.4599) -- (-0.8892,-1.4458) -- (-0.8664,-1.4319) -- (-0.8436,-1.4182) -- (-0.8208,-1.4047) -- (-0.7980,-1.3915) -- (-0.7752,-1.3786) -- (-0.7524,-1.3659) -- (-0.7296,-1.3535) -- (-0.7068,-1.3413) -- (-0.6840,-1.3295) -- (-0.6612,-1.3179) -- (-0.6384,-1.3066) -- (-0.6156,-1.2956) -- (-0.5928,-1.2849) -- (-0.5700,-1.2746) -- (-0.5472,-1.2645) -- (-0.5244,-1.2548) -- (-0.5016,-1.2455) -- (-0.4788,-1.2365) -- (-0.4560,-1.2278) -- (-0.4332,-1.2195) -- (-0.4104,-1.2116) -- (-0.3876,-1.2041) -- (-0.3648,-1.1969) -- (-0.3420,-1.1902) -- (-0.3192,-1.1838) -- (-0.2964,-1.1779) -- (-0.2736,-1.1724) -- (-0.2508,-1.1673) -- (-0.2280,-1.1626) -- (-0.2052,-1.1583) -- (-0.1824,-1.1545) -- (-0.1596,-1.1511) -- (-0.1368,-1.1482) -- (-0.1140,-1.1457) -- (-0.0912,-1.1436) -- (-0.0684,-1.1421) -- (-0.0456,-1.1409) -- (-0.0228,-1.1402) -- (0.0000,-1.1400) -- (0.0228,-1.1402) -- (0.0456,-1.1409) -- (0.0684,-1.1421) -- (0.0912,-1.1436) -- (0.1140,-1.1457) -- (0.1368,-1.1482) -- (0.1596,-1.1511) -- (0.1824,-1.1545) -- (0.2052,-1.1583) -- (0.2280,-1.1626) -- (0.2508,-1.1673) -- (0.2736,-1.1724) -- (0.2964,-1.1779) -- (0.3192,-1.1838) -- (0.3420,-1.1902) -- (0.3648,-1.1969) -- (0.3876,-1.2041) -- (0.4104,-1.2116) -- (0.4332,-1.2195) -- (0.4560,-1.2278) -- (0.4788,-1.2365) -- (0.5016,-1.2455) -- (0.5244,-1.2548) -- (0.5472,-1.2645) -- (0.5700,-1.2746) -- (0.5928,-1.2849) -- (0.6156,-1.2956) -- (0.6384,-1.3066) -- (0.6612,-1.3179) -- (0.6840,-1.3295) -- (0.7068,-1.3413) -- (0.7296,-1.3535) -- (0.7524,-1.3659) -- (0.7752,-1.3786) -- (0.7980,-1.3915) -- (0.8208,-1.4047) -- (0.8436,-1.4182) -- (0.8664,-1.4319) -- (0.8892,-1.4458) -- (0.9120,-1.4599) -- (0.9348,-1.4743) -- (0.9576,-1.4888) -- (0.9804,-1.5036) -- (1.0032,-1.5186) -- (1.0260,-1.5337) -- (1.0488,-1.5491) -- (1.0716,-1.5646) -- (1.0944,-1.5803) -- (1.1172,-1.5962) -- (1.1400,-1.6122) -- (1.1628,-1.6284) -- (1.1856,-1.6448) -- (1.2084,-1.6613) -- (1.2312,-1.6779) -- (1.2540,-1.6947) -- (1.2768,-1.7117) -- (1.2996,-1.7287) -- (1.3224,-1.7460) -- (1.3452,-1.7633) -- (1.3680,-1.7807) -- (1.3908,-1.7983) -- (1.4136,-1.8160) -- (1.4364,-1.8338) -- (1.4592,-1.8517) -- (1.4820,-1.8697) -- (1.5048,-1.8879) -- (1.5276,-1.9061) -- (1.5504,-1.9244) -- (1.5732,-1.9428) -- (1.5960,-1.9613) -- (1.6188,-1.9799) -- (1.6416,-1.9986) -- (1.6644,-2.0174) -- (1.6872,-2.0362) -- (1.7100,-2.0552) -- (1.7328,-2.0742) -- (1.7556,-2.0933) -- (1.7784,-2.1124) -- (1.8012,-2.1316) -- (1.8240,-2.1509) -- (1.8468,-2.1703) -- (1.8696,-2.1897) -- (1.8924,-2.2092) -- (1.9152,-2.2288) -- (1.9380,-2.2484) -- (1.9608,-2.2681) -- (1.9836,-2.2879) -- (2.0064,-2.3076) -- (2.0292,-2.3275) -- (2.0520,-2.3474) -- (2.0748,-2.3674) -- (2.0976,-2.3874) -- (2.1204,-2.4074) -- (2.1432,-2.4275) -- (2.1660,-2.4477) -- (2.1888,-2.4679) -- (2.2116,-2.4881) -- (2.2344,-2.5084) -- (2.2572,-2.5287) -- (2.2800,-2.5491) -- (2.3028,-2.5695) -- (2.3256,-2.5900) -- (2.3484,-2.6105) -- (2.3712,-2.6310) -- (2.3940,-2.6516) -- (2.4168,-2.6722) -- (2.4396,-2.6928) -- (2.4624,-2.7135) -- (2.4852,-2.7342) -- (2.5080,-2.7549) -- (2.5308,-2.7757) -- (2.5536,-2.7965) -- (2.5764,-2.8173) -- (2.5992,-2.8382) -- (2.6220,-2.8591) -- (2.6448,-2.8800) -- (2.6676,-2.9010) -- (2.6904,-2.9220) -- (2.7132,-2.9430) -- (2.7360,-2.9640) -- (2.7588,-2.9851) -- (2.7816,-3.0061) -- (2.8044,-3.0273) -- (2.8272,-3.0484) -- (2.8500,-3.0695) -- (2.8728,-3.0907) -- (2.8956,-3.1119) -- (2.9184,-3.1332) -- (2.9412,-3.1544) -- (2.9640,-3.1757) -- (2.9868,-3.1970) -- (3.0096,-3.2183) -- (3.0324,-3.2396) -- (3.0552,-3.2610) -- (3.0780,-3.2823) -- (3.1008,-3.3037) -- (3.1236,-3.3251) -- (3.1464,-3.3466) -- (3.1692,-3.3680) -- (3.1920,-3.3895) -- (3.2148,-3.4109) -- (3.2376,-3.4324) -- (3.2604,-3.4540) -- (3.2832,-3.4755) -- (3.3060,-3.4970) -- (3.3288,-3.5186) -- (3.3516,-3.5402) -- (3.3744,-3.5618) -- (3.3972,-3.5834) -- (3.4200,-3.6050)}
\def\bAsymR{(3.7050,-3.7050)}
\def\bAsymL{(-3.7050,-3.7050)}
\def\bTanAx{-1.0902}
\def\bTanAy{-0.6088}
\def\bTanBx{2.0852}
\def\bTanBy{-1.8789}
\def\bParAx{-1.3027}
\def\bParAy{0.2931}
\def\bParBx{1.8727}
\def\bParBy{-0.9771}
\def\bFrame{3.0210}
\def\bRunX{2.9868}
\def\bRunY{-3.1970}
\def\bLift{2.2575}
\def\capY{-2.7950}
\def\aHalfLine{2.3650}

  %%--------------------------------------------------  (a) the free ball
  \begin{scope}
    \draw[gray!45,line width=0.5pt] (-\aR-0.45,0) -- (\aR+0.45,0);
    \draw[line width=0.9pt] (0,0) circle (\aR);
    \draw[line width=0.7pt,gray!70] (\aTanR,-\aHalfLine) -- (\aTanR,\aHalfLine);
    \draw[line width=0.7pt,gray!70] (\aTanL,-\aHalfLine) -- (\aTanL,\aHalfLine);
    \draw[line width=0.9pt,|-|] (\aAx,0.40) -- (\aTanR,0.40);
    \draw[line width=0.9pt,|-|] (\aTanL,-0.40) -- (\aAx,-0.40);
    \fill (0,0) circle (1.5pt);
    \node[below left,font=\footnotesize,inner sep=1.5pt] at (0,0) {$O$};
    \fill (\aAx,0) circle (1.5pt);
    \node[below right,font=\footnotesize,inner sep=1.5pt] at (\aAx,0) {$A$};
    \node[above right,font=\scriptsize,inner sep=1pt] at (\aTanR,0.42) {$v-A(\hat y)$};
    \node[below,font=\scriptsize] at (\aMidL,-0.42) {$v+A(\hat y)$};
    \draw[->,line width=0.7pt] (-0.25,\aHalfLine+0.30) -- (0.55,\aHalfLine+0.30);
    \node[right,font=\scriptsize,inner sep=2pt] at (0.55,\aHalfLine+0.30) {$\hat y$};
    \node[font=\footnotesize] at (0,\capY) {(a)\ \ $\epsilon=+1$, the free ball};
  \end{scope}

  %%--------------------------------------------------  (b) the mass shell
  \begin{scope}[xshift=8.55cm,yshift=\bLift cm]
    \draw[gray!55,dotted,line width=0.8pt] (0,0) -- \bAsymR;
    \draw[gray!55,dotted,line width=0.8pt] (0,0) -- \bAsymL;
    \draw[line width=0.9pt] \bsheet;
    \draw[line width=0.7pt,gray!70] (\bTanAx,\bTanAy) -- (\bTanBx,\bTanBy);
    \draw[line width=0.7pt,gray!70,densely dashed] (\bParAx,\bParAy) -- (\bParBx,\bParBy);
    \draw[line width=0.9pt,|-|] (\bAX,\bAY) -- (\bLegX,\bLegY);
    \fill (\bLegX,\bLegY) circle (1.5pt);
    \node[font=\scriptsize,anchor=north west] at ($(\bLegX,\bLegY)+(0.10,-0.16)$) {$P_{*}$};
    \fill (\bAX,\bAY) circle (1.5pt);
    \node[above,font=\footnotesize,inner sep=2pt] at (\bAX,\bAY) {$A$};
    \node[left,font=\scriptsize,inner sep=3pt] at (\bAX,\bAY-0.42) {$v\,L(y)$};
    \draw[->,line width=0.7pt] (\bRunX,\bRunY) -- ++(0.42,-0.42);
    \node[font=\scriptsize,anchor=north east,align=right] at (\bRunX+0.30,\bRunY-0.42)
    {sheet unbounded,\\no supporting line};
  \end{scope}
  \node[font=\footnotesize] at (8.55,\capY) {(b)\ \ $\epsilon=-1$, the mass shell};
\end{tikzpicture}

  \caption{The support-function reading of the Randers function, under both
  signatures. In~(a) the free level is the ball of $g$-radius $v$, and the
  supporting lines with normals $\pm\hat y$ stand at $v \mp A(\hat y)$ from the
  potential $A$ where they stand at $v$ from the centre $O$; that displacement of
  the centre is the field's whole contribution, and it is what makes the length
  depend upon direction. In~(b) the level is the mass shell, whose sheet is
  unbounded, so that no supporting line exists in the opposite sense and the
  supremum~\eqref{eq.support} is lost. The infimum~\eqref{eq.supportlorentz} remains, attained at the Legendre
  image~\eqref{eq.legendreimage}, $P_{*} = v\,y^\flat/\vert y\vert_g$, and
  separating $A$ from the supporting line by $vL(y)$. The dotted lines are the
  light cone. Both panels are drawn in a fibre of $T^*M$ in an orthonormal
  coframe.}
  \label{fig.support}
\end{figure}
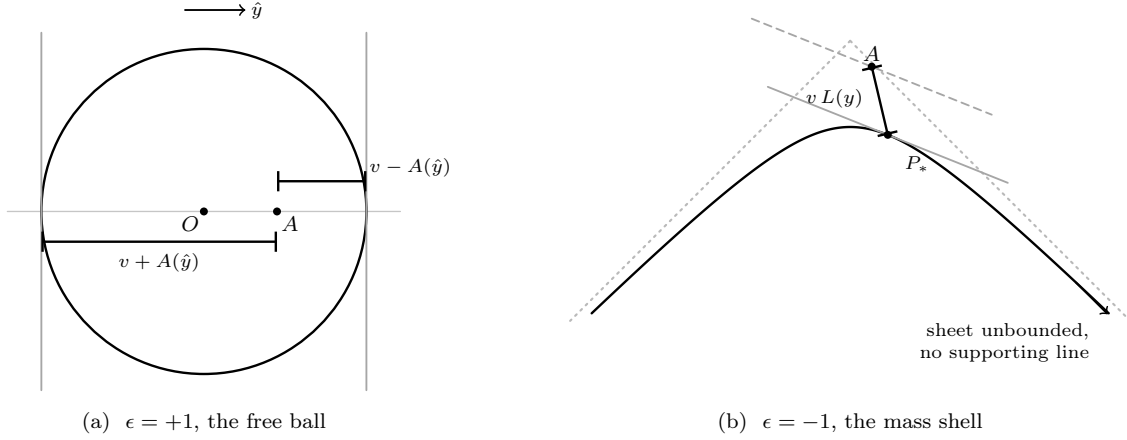

The free ball, the appeal to Cauchy--Schwarz~\eqref{eq.cauchyschwarz} and the displaced centre go together. The construction holds under either signature, though the picture is a Riemannian one, with the level~\eqref{eq.freelevel} become a mass shell. That is the separation of the length from the dynamics which Section~\ref{sec.indefinite} records, met here from the side of the phase space.

Through their quotient, the invariant fixed by the affine reparametrisations above, the metric~\eqref{eq.randers} combines the field and the speed. Thus, for every value of that quotient the connection is replaced by a Finsler metric, held apart from its neighbours by the symmetry of inertial motion.

That the trajectories are the geodesics of such a metric is classical. By Maupertuis' principle an exact magnetic Lagrangian at fixed energy is sent to a length functional whose integrand includes a term of degree one in the velocity. Above the Ma\~n\'e critical value the flow on a fixed energy level is a Randers geodesic flow~\cite{mane1997,contreras1998}. On the symplectic side the same value governs, since the level is of contact type above it and, on a surface, only there~\cite{contreras1998,cieliebak2010}. It is the infimum, over the primitives of the field, of half the squared supremum of their length --- which for the primitive in hand is the threshold we found above, stated as a condition on the energy. Randers himself introduced his metric so that the trajectories of a charge should be
geodesics~\cite{randers1941}. The correspondence has been developed for stationary
media~\cite{gibbons2009} and for magnetic billiards~\cite{tabachnikov2004}, where the energy
dependence of the one-form is recorded as a feature of it.

Let us be exact about what of this we claim. The correspondence itself we claim nothing of, nor
the Ma\~n\'e value, nor the contact type of the levels above it, nor the energy dependence as a
recorded feature. All four are the work of the authors just named. What is ours
begins with the identity~\eqref{eq.hilbert}, that the primitive $\lambda_A$ of the twisted form,
restricted to the free level, is $v$ times the Hilbert form of the Randers metric. The second is
the Liouville field~\eqref{eq.Y}, the fibrewise dilation centred at the potential, which is what puts that one-form there. Ours too is the support function~\eqref{eq.support}, whose centre is
displaced to the potential, together with its Lorentzian companion~\eqref{eq.supportlorentz}. So the metric is taken from the primitive of the structure
the flow is already in, rather than built beside it.

One further thing is ours, and it is a reason rather than an object. That the one-form depends on
the energy is recorded in the literature; why it must is not. The symmetry which selects the
metric moves the field and the speed together and fixes their quotient, so a trajectory answers to
that quotient and the one-form takes the speed along with it.

%%------------------------------------------------------------
\subsection{The nonlinear connection}\label{sec.transport}
%%------------------------------------------------------------
\emph{What is left of parallel transport, once the metric has gone?}

A law of transport remains, albeit not one of the affine kind. Let us return to the pair inertial motion requires, only one half of which has been supplied, and to the horizontal distribution above, which need not be linear on each fibre --- a further demand no principle has made. That demand we drop, to see what remains.

The object which replaces the connection is familiar in a form which does not look familiar, so let us have it from the case we know. For a metric $g$ the geodesic equation may be written
\beq \label{eq.sprayform}
\ddot\gamma^i + 2 G^i\!\left(\gamma,\dot\gamma\right) = 0
\quad \text{with} \quad
G^i = \tfrac12\,\mathring{\Gamma}^i{}_{jk}\,\dot\gamma^j\dot\gamma^k ,
\eeq
where the coefficients $G^i$ are quadratic in the velocity, being built from the Christoffel symbols of $g$. A function which is not the length of a metric gives an equation of that same shape with coefficients which need not be quadratic, homogeneous of degree two in the velocity without being quadratic in it, and such a collection of coefficients is called a \emph{spray}. We have one for every Finsler function, obtained from it as the geodesic equation is obtained from a metric.

We obtain the connection from the spray by differentiating along the fibre. Writing
\beq \label{eq.nlcoeff}
N^i{}_j = \frac{\partial G^i}{\partial y^j}
\eeq
we obtain the coefficients of a \emph{nonlinear connection} on the slit tangent bundle, homogeneous of degree one there. That definition is the familiar one, as the case we came from shows. Differentiating the quadratic spray of the geodesic equation~\eqref{eq.sprayform} returns $\mathring{\Gamma}^i{}_{jk}\,y^k$, which is the object parallel transport has always used, and the two agree because the Levi-Civita connection is symmetric in its lower pair. Thus, we have in the fibre derivative~\eqref{eq.nlcoeff} the old rule written in a form which holds once quadraticity is lost.

Such a connection does for us what the affine one does. It says which curves in the tangent bundle are to be counted as horizontal, hence which vectors are carried along a motion without turning, and its horizontal curves are the geodesics of the function it came from. In this sense a Finsler metric determines a law of transport as surely as a Riemannian one does, though the rule it determines is not linear on the fibres.

For the metric~\eqref{eq.randers} the spray is explicit, and the connection with it. Writing $\mathring{G}^i$ for the spray of $g$,
\beq \label{eq.spray}
G^i = \mathring{G}^i
 + \frac{\vert y\vert_g}{2v}\,g^{ik}\left(\d A\right)_{kj} y^j
 + \rho\,y^i,
\eeq
with $\rho$ homogeneous of degree one in the velocity~\cite{baochernshen2000}. We have three terms, and each of them says something to us. The first is the geodesic spray of $g$, which is what a neutral particle answers to. The second is the Lorentz force of $\d A$ multiplied by the speed, and the factor of the speed is what makes it homogeneous of degree two overall, as a spray must be, whilst leaving it short of quadratic. That is the freedom an affine connection does not have, the degree count of the Randers morphism met once more, on the other bundle. The third points along the velocity, so it changes the rate at which a trajectory is traversed and leaves the trajectory where it was.

Let us be careful about that third term, since a term left over in an equation of motion is one we must explain. On a motion of speed $v$ the horizontality condition
\beq \label{eq.horizontality}
\ddot\gamma^i + 2G^i = 0
\eeq
returns the law~\eqref{eq.emlaw} up to that tangential term. Indeed, a term along the velocity cannot turn a curve, so the horizontal curves and the trajectories of a charge are the same curves differently parametrised, and as paths in the configuration space they coincide. With $\d A$ the middle term vanishes as well, hence a closed potential returns the geodesics of $g$ and the trajectories record the field.

The neutral particle we now put to the transport, as we have already put it to the structure, and it serves as a second check on the whole arrangement. With the two constitutive quantities restored, the potential enters the metric~\eqref{eq.randers} as $-qA/mv$ and the field the spray~\eqref{eq.spray} as $-q\,\d A/m$, with $A$ there a primitive of the field itself. At vanishing charge the metric returns the length of $g$, both added terms of the decomposition~\eqref{eq.spray} vanish and the coefficients reduce to
\beq \label{eq.levicivita}
N^i{}_j = \mathring{\Gamma}^i{}_{jk}\,y^k,
\eeq
linear in the velocity, with the Christoffel symbols of $g$ for its coefficients. That is the connection of Levi-Civita, which is where we began. Moreover, the departure from it is of first order in the charge, with the Lorentz force multiplied by the speed for its leading term, though the spray is quadratic in the velocity at zero charge. Thus, the neutral member of the family of transports is the affine one, precisely as the neutral member of the line of kinematic structures is the canonical one. Indeed, the two are one statement told on the two bundles. A charge moves a transport out of the affine class.

With the linearity we lose the algebra of the tangent spaces, and we lose one half of that algebra whilst keeping the other. Under $N^i{}_j$ the transport still commutes with scaling, since the coefficients are homogeneous of degree one in the velocity, so a particle at twice the speed is carried along a scaled solution of the same equation. However, addition is another matter. In the affine case the coefficients are linear in the velocity and the transport is an isomorphism of tangent spaces, which is what lets a frame be carried along a curve and its vectors be added at the far end. Here the transport is positively homogeneous, without being linear, so a frame carried by it no longer superposes. Thus, two vectors may each be transported, whilst their sum is not the transport of their sum.

Under reversal the nonlinearity is felt again. In the coefficients~\eqref{eq.nlcoeff} the Christoffel part is odd in the velocity and the field part, having $\vert y\vert_g$ for a factor, is even. Reversing the direction of travel therefore does not invert the transport, an asymmetry the Randers length shows as well in $F_{\rm R}(-y) \neq F_{\rm R}(y)$. Each particle meets it scaled by the quotient $qF/mv$, into which its constitution enters through $q/m$.

Ingarden proposed this construction, pairing a Randers metric with a nonlinear connection carrying the field~\cite{ingarden1976}, and Miron developed it~\cite{mironanastasiei1994}. There the autoparallel curves of the canonical nonlinear connection of the Lagrange space of electrodynamics are the Lorentz equations. That literature holds two such connections apart and treats them as different objects. In a \emph{Randers space} the metric comes with its own Cartan nonlinear connection, whilst in an \emph{Ingarden space} it comes with the Lorentz one, whose field term is additive, of degree zero, correcting the Christoffel part by $+\,g^{ik}F_{kj}$. The decomposition~\eqref{eq.spray} is the first of those two, written out.

Two senses of the word affine are met here, since as unparametrised paths the trajectories are the geodesics of no affine connection, a statement which holds only where the field is non-trivial. As a connection the spray fails to be quadratic already when the one-form is merely non-parallel, whether or not the field vanishes. The first is a statement about the field, the second about the potential.

We have then a family of structures --- a Finsler metric and a transport for every value of the quotient, each settling what inertial motion is for a particle of that value. It is worth asking what such a family amounts to.

%%============================================================
\section{Two fields on the round sphere}\label{sec.sphere}
%%============================================================
Let us carry the construction through on one configuration, and twice over. Take $M = \Sph^2$ with the round metric throughout and put on it two fields which differ in one respect --- the flux each sends through the sphere. Everything the construction has to say about them follows from that difference, and between them we shall find the two halves of the pair inertial motion requires coming apart. Throughout, $\theta$ is the colatitude and $\varphi$ the longitude, so that
\beq \label{eq.roundsphere}
g = \d\theta^2 + \sin^2\!\theta\ \d\varphi^2
\quad \text{and} \quad
{\rm Vol}_g = \sin\theta\ \d\theta\wedge\d\varphi .
\eeq

%%------------------------------------------------------------
\subsection{The uniform ambient field}\label{sec.uniform}
%%------------------------------------------------------------
Embed the sphere in $\Reals^3$ in the usual way and let the field on it be the restriction of a uniform ambient one of strength $B$ along the polar axis. That ambient field is $B\,\d x\wedge\d y$, whose pull-back to the sphere is the polar component of it, so
\beq \label{eq.spherefield}
F = B\cos\theta\ {\rm Vol}_g
\quad \text{and} \quad
A = \frac{B}{2}\sin^2\!\theta\ \d\varphi,
\quad \text{with} \quad
\alpha := \frac{qB}{2m},
\eeq
the potential being the one whose exterior derivative returns the field,
\beq \label{eq.spherecheck}
\d A = B\sin\theta\cos\theta\ \d\theta\wedge\d\varphi = B\cos\theta\ {\rm Vol}_g .
\eeq
Here the field is a \emph{function} multiple of the area form, whose sign changes across the equator. The flux therefore cancels between the hemispheres,
\beq \label{eq.fluxzero}
\oint_{\Sph^2} F = B\int_0^{2\pi}\!\!\int_0^{\pi} \cos\theta\,\sin\theta\ \d\theta\,\d\varphi
= 2\pi B\int_0^\pi \tfrac12\sin 2\theta\ \d\theta = 0 ,
\eeq
which is a topological statement about this configuration. As much flux is pushed through the northern hemisphere by a uniform ambient field as is drawn back through the southern one. The sphere, having no boundary, retains none of it.

We may therefore return to Section~\ref{sec.sym}. Since none of the charge is retained by the structure, every pair of charges is related by a fibre translation, the one built from $A$. Moreover, the potential is smooth at both poles, where $\d\varphi$ is not, because $\sin^2\!\theta$ vanishes there to second order, killing the singularity of the coordinate one-form.

By the Hilbert-form identification~\eqref{eq.hilbert} the free level of a particle of charge $q$ and mass $m$ at speed $v$ has the Randers metric
\beq \label{eq.spheremetric}
F_{\rm R}(y) = \vert y\vert_g + \frac{\alpha}{v}\sin^2\!\theta\ \d\varphi(y),
\eeq
whose one-form we may measure at once. By the declaration~\eqref{eq.constitutive} that one-form is $-(q/m)A$, which leaves the length below untouched. The round metric~\eqref{eq.roundsphere} has an inverse which is diagonal with entries $1$ and $\sin^{-2}\!\theta$, and the effective potential $(q/m)A$ has the single component $\alpha\sin^2\!\theta$ along $\d\varphi$, so
\beq \label{eq.spherenorm}
\left\vert \frac{q}{m}A \right\vert_g^2
= g^{\varphi\varphi}\left( \alpha\sin^2\!\theta \right)^2
= \frac{\alpha^2\sin^4\!\theta}{\sin^2\!\theta}
= \alpha^2\sin^2\!\theta ,
\quad \text{whence} \quad
\frac{\vert A\vert_g}{v} = \frac{\alpha}{v}\,\sin\theta .
\eeq
We find the supremum of $\sin\theta$ over the sphere to be unity, reached on the equator, where the parallels are longest and the most room is left for the potential to grow. Hence the Randers condition of Section~\ref{sec.finsler} becomes $\alpha/v < 1$, which falls at the threshold
\beq \label{eq.threshold}
\alpha = v ,
\quad \text{that is} \quad
\frac{qB}{2m} = v .
\eeq
Below that threshold, at every value of the quotient $\alpha/v$, we obtain the trajectories as the geodesics of the metric~\eqref{eq.spheremetric}. We may then appeal to the known criterion under which one of them closes~\cite{lm2026closure}.

There the general construction stops, while a particular geometry starts. The metric~\eqref{eq.spheremetric} comes from a navigation problem in which the wind is not Killing and the flag curvature not constant. The equator is a closed geodesic whose stability changes at half the speed, while the closed geodesics through the poles stay degenerate at every field strength. Of that we take up nothing here. None of it follows from the framework above.

%%------------------------------------------------------------
\subsection{The monopole}\label{sec.monopole}
%%------------------------------------------------------------
Where the flux does not cancel the construction answers differently at every step. Let us keep the sphere with its metric, taking the field of a Dirac monopole of strength $k$, which on the sphere is the area form itself,
\beq \label{eq.monopolefield}
F = k\ {\rm Vol}_g = k\sin\theta\ \d\theta\wedge\d\varphi
\quad \text{and} \quad
\oint_{\Sph^2} F = k\int_0^{2\pi}\!\!\int_0^\pi \sin\theta\ \d\theta\,\d\varphi = 4\pi k ,
\eeq
every flux here being taken over the whole sphere. The field is closed and source-free, whilst the flux~\eqref{eq.monopolefield} does not vanish, the sphere retaining what a uniform field returns to it.

No potential covers the sphere. It is worth seeing in what manner the attempt fails, since the failure decides the rest of the subsection. It is covered by two charts, the northern omitting the south pole, the southern omitting the north, and on them we have
\beq \label{eq.wuyang}
A_{\rm N} = k\left( 1 - \cos\theta \right)\d\varphi
\quad \text{and} \quad
A_{\rm S} = -k\left( 1 + \cos\theta \right)\d\varphi ,
\eeq
each of which returns the field~\eqref{eq.monopolefield} upon differentiation, differing on the overlap by
\beq \label{eq.wuyangdiff}
A_{\rm N} - A_{\rm S} = 2k\ \d\varphi ,
\eeq
a closed one-form which is not the differential of any function there. The $\d$ in its name makes no claim of exactness. On the overlap, which is the sphere less both poles, the longitude is no function but jumps by $2\pi$ upon return, so that $\d\varphi$ has the period $\oint\d\varphi = 2\pi$ around every parallel and no primitive with it. That is the description Dirac gave the problem and the one from which quantisation is usually reached~\cite{trautman1977,ryder1980}. We have integrated the trajectories in this field elsewhere, using that treatment here as given~\cite{lm2025monopole}.

The construction of Section~\ref{sec.kin} holds intact. The twisted family~\eqref{eq.family} requires the field to be closed and global --- as it is --- whilst requiring nothing of its flux. A charge still selects its member of the line, whose trajectories are still inertial motions of the kinetic energy.

Moreover, the charge stops being a matter of description. Where the flux cancels, a fibre translation carries the kinematics of any charge to that of any other, the two being told apart by their energies. Here the class identity of Section~\ref{sec.sym} gives $q'[F] = q\,\deg(\phi)[F]$ in $H^2(\Sph^2;\Reals)$, hence an orientation-preserving diffeomorphism leaves $q' = q$, with no translation relating two charges. In this sense the flux makes the charge a property of the particle.

It is worth marking what this is not. Healey asks which quantity of the electromagnetic field is real, answering with the holonomies of loops~\cite{healey2001}. The monopole is that point at its strongest, since a potential is no globally defined thing. We ask a narrower question with a sharper answer, whether two charges are related by a symmetry of the structure a charge selects. That question is answered by the flux. The two questions meet on one mathematical fact and part on everything else.

Finally the two halves of the pair come apart. The form fails, twice over. With no global primitive there is no $\lambda_A$, no dilation centred at a potential, no Randers metric, the whole of Section~\ref{sec.finsler} resting on a one-form which this field does not admit. Nor does working in one chart at a time save anything, since the length of the chart potential
\beq \label{eq.wuyangnorm}
\left\vert A_{\rm N} \right\vert_g
= \frac{k\left( 1 - \cos\theta \right)}{\sin\theta}
= k\,\tan\frac{\theta}{2}
\eeq
grows without bound as $\theta$ approaches the pole that chart omits. The Randers condition $\sup_M \vert A\vert_g / v < 1$ therefore fails in either chart at every charge, however small, where in the uniform field~\eqref{eq.spherefield} it failed only above a threshold (Figure~\ref{fig.randers}).

%%------------------------------------------------------------
%%  Figure --- drawn in TikZ, every coordinate a literal.
%%  Inlined 2026-09-10 so this manuscript is a single source file.
%%------------------------------------------------------------
%%  The float, ready to drop into Section 5(b) after the display eq.wuyangnorm.
%%  Requires \usepackage{tikz} in the preamble and nothing else.
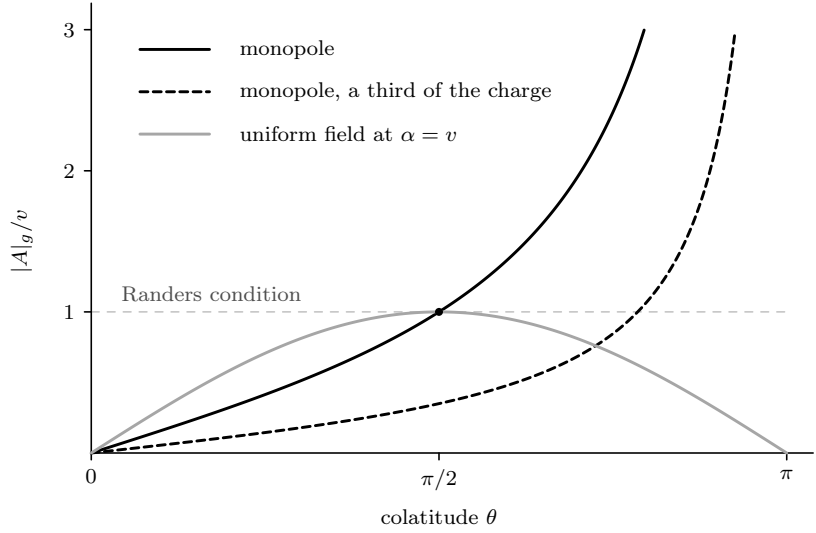
\begin{figure}[htbp]
  \centering
  %%  Section 5: the Randers condition in the two fields on the round sphere.
\begin{tikzpicture}[x=1cm,y=1cm,line join=round,line cap=round]
  \def\uniformpath{(0.0000,0.0000) -- (0.0230,0.0147) -- (0.0460,0.0293) -- (0.0690,0.0440) -- (0.0920,0.0586) -- (0.1150,0.0733) -- (0.1380,0.0879) -- (0.1610,0.1026) -- (0.1840,0.1172) -- (0.2070,0.1318) -- (0.2300,0.1465) -- (0.2530,0.1611) -- (0.2760,0.1757) -- (0.2990,0.1903) -- (0.3220,0.2048) -- (0.3450,0.2194) -- (0.3680,0.2340) -- (0.3910,0.2485) -- (0.4140,0.2630) -- (0.4370,0.2775) -- (0.4600,0.2920) -- (0.4830,0.3065) -- (0.5060,0.3209) -- (0.5290,0.3354) -- (0.5520,0.3498) -- (0.5750,0.3642) -- (0.5980,0.3785) -- (0.6210,0.3929) -- (0.6440,0.4072) -- (0.6670,0.4215) -- (0.6900,0.4358) -- (0.7130,0.4500) -- (0.7360,0.4642) -- (0.7590,0.4784) -- (0.7820,0.4926) -- (0.8050,0.5067) -- (0.8280,0.5208) -- (0.8510,0.5348) -- (0.8740,0.5489) -- (0.8970,0.5629) -- (0.9200,0.5768) -- (0.9430,0.5908) -- (0.9660,0.6046) -- (0.9890,0.6185) -- (1.0120,0.6323) -- (1.0350,0.6461) -- (1.0580,0.6598) -- (1.0810,0.6735) -- (1.1040,0.6872) -- (1.1270,0.7008) -- (1.1500,0.7143) -- (1.1730,0.7279) -- (1.1960,0.7413) -- (1.2190,0.7548) -- (1.2420,0.7682) -- (1.2650,0.7815) -- (1.2880,0.7948) -- (1.3110,0.8080) -- (1.3340,0.8212) -- (1.3570,0.8344) -- (1.3800,0.8474) -- (1.4030,0.8605) -- (1.4260,0.8735) -- (1.4490,0.8864) -- (1.4720,0.8993) -- (1.4950,0.9121) -- (1.5180,0.9249) -- (1.5410,0.9376) -- (1.5640,0.9502) -- (1.5870,0.9628) -- (1.6100,0.9753) -- (1.6330,0.9878) -- (1.6560,1.0002) -- (1.6790,1.0126) -- (1.7020,1.0248) -- (1.7250,1.0371) -- (1.7480,1.0492) -- (1.7710,1.0613) -- (1.7940,1.0733) -- (1.8170,1.0853) -- (1.8400,1.0972) -- (1.8630,1.1090) -- (1.8860,1.1208) -- (1.9090,1.1325) -- (1.9320,1.1441) -- (1.9550,1.1556) -- (1.9780,1.1671) -- (2.0010,1.1785) -- (2.0240,1.1899) -- (2.0470,1.2011) -- (2.0700,1.2123) -- (2.0930,1.2234) -- (2.1160,1.2344) -- (2.1390,1.2454) -- (2.1620,1.2563) -- (2.1850,1.2671) -- (2.2080,1.2778) -- (2.2310,1.2885) -- (2.2540,1.2990) -- (2.2770,1.3095) -- (2.3000,1.3199) -- (2.3230,1.3303) -- (2.3460,1.3405) -- (2.3690,1.3507) -- (2.3920,1.3607) -- (2.4150,1.3707) -- (2.4380,1.3806) -- (2.4610,1.3905) -- (2.4840,1.4002) -- (2.5070,1.4099) -- (2.5300,1.4194) -- (2.5530,1.4289) -- (2.5760,1.4383) -- (2.5990,1.4476) -- (2.6220,1.4568) -- (2.6450,1.4659) -- (2.6680,1.4750) -- (2.6910,1.4839) -- (2.7140,1.4927) -- (2.7370,1.5015) -- (2.7600,1.5102) -- (2.7830,1.5187) -- (2.8060,1.5272) -- (2.8290,1.5356) -- (2.8520,1.5439) -- (2.8750,1.5521) -- (2.8980,1.5602) -- (2.9210,1.5682) -- (2.9440,1.5761) -- (2.9670,1.5839) -- (2.9900,1.5916) -- (3.0130,1.5992) -- (3.0360,1.6067) -- (3.0590,1.6141) -- (3.0820,1.6214) -- (3.1050,1.6287) -- (3.1280,1.6358) -- (3.1510,1.6428) -- (3.1740,1.6497) -- (3.1970,1.6565) -- (3.2200,1.6632) -- (3.2430,1.6698) -- (3.2660,1.6763) -- (3.2890,1.6827) -- (3.3120,1.6890) -- (3.3350,1.6952) -- (3.3580,1.7013) -- (3.3810,1.7073) -- (3.4040,1.7131) -- (3.4270,1.7189) -- (3.4500,1.7246) -- (3.4730,1.7301) -- (3.4960,1.7356) -- (3.5190,1.7409) -- (3.5420,1.7462) -- (3.5650,1.7513) -- (3.5880,1.7563) -- (3.6110,1.7612) -- (3.6340,1.7660) -- (3.6570,1.7707) -- (3.6800,1.7753) -- (3.7030,1.7798) -- (3.7260,1.7841) -- (3.7490,1.7884) -- (3.7720,1.7925) -- (3.7950,1.7966) -- (3.8180,1.8005) -- (3.8410,1.8043) -- (3.8640,1.8080) -- (3.8870,1.8116) -- (3.9100,1.8151) -- (3.9330,1.8185) -- (3.9560,1.8217) -- (3.9790,1.8249) -- (4.0020,1.8279) -- (4.0250,1.8308) -- (4.0480,1.8336) -- (4.0710,1.8363) -- (4.0940,1.8389) -- (4.1170,1.8413) -- (4.1400,1.8437) -- (4.1630,1.8459) -- (4.1860,1.8480) -- (4.2090,1.8501) -- (4.2320,1.8519) -- (4.2550,1.8537) -- (4.2780,1.8554) -- (4.3010,1.8569) -- (4.3240,1.8584) -- (4.3470,1.8597) -- (4.3700,1.8609) -- (4.3930,1.8620) -- (4.4160,1.8630) -- (4.4390,1.8638) -- (4.4620,1.8646) -- (4.4850,1.8652) -- (4.5080,1.8657) -- (4.5310,1.8661) -- (4.5540,1.8664) -- (4.5770,1.8666) -- (4.6000,1.8667) -- (4.6230,1.8666) -- (4.6460,1.8664) -- (4.6690,1.8661) -- (4.6920,1.8657) -- (4.7150,1.8652) -- (4.7380,1.8646) -- (4.7610,1.8638) -- (4.7840,1.8630) -- (4.8070,1.8620) -- (4.8300,1.8609) -- (4.8530,1.8597) -- (4.8760,1.8584) -- (4.8990,1.8569) -- (4.9220,1.8554) -- (4.9450,1.8537) -- (4.9680,1.8519) -- (4.9910,1.8501) -- (5.0140,1.8480) -- (5.0370,1.8459) -- (5.0600,1.8437) -- (5.0830,1.8413) -- (5.1060,1.8389) -- (5.1290,1.8363) -- (5.1520,1.8336) -- (5.1750,1.8308) -- (5.1980,1.8279) -- (5.2210,1.8249) -- (5.2440,1.8217) -- (5.2670,1.8185) -- (5.2900,1.8151) -- (5.3130,1.8116) -- (5.3360,1.8080) -- (5.3590,1.8043) -- (5.3820,1.8005) -- (5.4050,1.7966) -- (5.4280,1.7925) -- (5.4510,1.7884) -- (5.4740,1.7841) -- (5.4970,1.7798) -- (5.5200,1.7753) -- (5.5430,1.7707) -- (5.5660,1.7660) -- (5.5890,1.7612) -- (5.6120,1.7563) -- (5.6350,1.7513) -- (5.6580,1.7462) -- (5.6810,1.7409) -- (5.7040,1.7356) -- (5.7270,1.7301) -- (5.7500,1.7246) -- (5.7730,1.7189) -- (5.7960,1.7131) -- (5.8190,1.7073) -- (5.8420,1.7013) -- (5.8650,1.6952) -- (5.8880,1.6890) -- (5.9110,1.6827) -- (5.9340,1.6763) -- (5.9570,1.6698) -- (5.9800,1.6632) -- (6.0030,1.6565) -- (6.0260,1.6497) -- (6.0490,1.6428) -- (6.0720,1.6358) -- (6.0950,1.6287) -- (6.1180,1.6214) -- (6.1410,1.6141) -- (6.1640,1.6067) -- (6.1870,1.5992) -- (6.2100,1.5916) -- (6.2330,1.5839) -- (6.2560,1.5761) -- (6.2790,1.5682) -- (6.3020,1.5602) -- (6.3250,1.5521) -- (6.3480,1.5439) -- (6.3710,1.5356) -- (6.3940,1.5272) -- (6.4170,1.5187) -- (6.4400,1.5102) -- (6.4630,1.5015) -- (6.4860,1.4927) -- (6.5090,1.4839) -- (6.5320,1.4750) -- (6.5550,1.4659) -- (6.5780,1.4568) -- (6.6010,1.4476) -- (6.6240,1.4383) -- (6.6470,1.4289) -- (6.6700,1.4194) -- (6.6930,1.4099) -- (6.7160,1.4002) -- (6.7390,1.3905) -- (6.7620,1.3806) -- (6.7850,1.3707) -- (6.8080,1.3607) -- (6.8310,1.3507) -- (6.8540,1.3405) -- (6.8770,1.3303) -- (6.9000,1.3199) -- (6.9230,1.3095) -- (6.9460,1.2990) -- (6.9690,1.2885) -- (6.9920,1.2778) -- (7.0150,1.2671) -- (7.0380,1.2563) -- (7.0610,1.2454) -- (7.0840,1.2344) -- (7.1070,1.2234) -- (7.1300,1.2123) -- (7.1530,1.2011) -- (7.1760,1.1899) -- (7.1990,1.1785) -- (7.2220,1.1671) -- (7.2450,1.1556) -- (7.2680,1.1441) -- (7.2910,1.1325) -- (7.3140,1.1208) -- (7.3370,1.1090) -- (7.3600,1.0972) -- (7.3830,1.0853) -- (7.4060,1.0733) -- (7.4290,1.0613) -- (7.4520,1.0492) -- (7.4750,1.0371) -- (7.4980,1.0248) -- (7.5210,1.0126) -- (7.5440,1.0002) -- (7.5670,0.9878) -- (7.5900,0.9753) -- (7.6130,0.9628) -- (7.6360,0.9502) -- (7.6590,0.9376) -- (7.6820,0.9249) -- (7.7050,0.9121) -- (7.7280,0.8993) -- (7.7510,0.8864) -- (7.7740,0.8735) -- (7.7970,0.8605) -- (7.8200,0.8474) -- (7.8430,0.8344) -- (7.8660,0.8212) -- (7.8890,0.8080) -- (7.9120,0.7948) -- (7.9350,0.7815) -- (7.9580,0.7682) -- (7.9810,0.7548) -- (8.0040,0.7413) -- (8.0270,0.7279) -- (8.0500,0.7143) -- (8.0730,0.7008) -- (8.0960,0.6872) -- (8.1190,0.6735) -- (8.1420,0.6598) -- (8.1650,0.6461) -- (8.1880,0.6323) -- (8.2110,0.6185) -- (8.2340,0.6046) -- (8.2570,0.5908) -- (8.2800,0.5768) -- (8.3030,0.5629) -- (8.3260,0.5489) -- (8.3490,0.5348) -- (8.3720,0.5208) -- (8.3950,0.5067) -- (8.4180,0.4926) -- (8.4410,0.4784) -- (8.4640,0.4642) -- (8.4870,0.4500) -- (8.5100,0.4358) -- (8.5330,0.4215) -- (8.5560,0.4072) -- (8.5790,0.3929) -- (8.6020,0.3785) -- (8.6250,0.3642) -- (8.6480,0.3498) -- (8.6710,0.3354) -- (8.6940,0.3209) -- (8.7170,0.3065) -- (8.7400,0.2920) -- (8.7630,0.2775) -- (8.7860,0.2630) -- (8.8090,0.2485) -- (8.8320,0.2340) -- (8.8550,0.2194) -- (8.8780,0.2048) -- (8.9010,0.1903) -- (8.9240,0.1757) -- (8.9470,0.1611) -- (8.9700,0.1465) -- (8.9930,0.1318) -- (9.0160,0.1172) -- (9.0390,0.1026) -- (9.0620,0.0879) -- (9.0850,0.0733) -- (9.1080,0.0586) -- (9.1310,0.0440) -- (9.1540,0.0293) -- (9.1770,0.0147) -- (9.2000,0.0000)}
\def\monopolepath{(0.0000,0.0000) -- (0.0230,0.0073) -- (0.0460,0.0147) -- (0.0690,0.0220) -- (0.0920,0.0293) -- (0.1150,0.0367) -- (0.1380,0.0440) -- (0.1610,0.0513) -- (0.1840,0.0587) -- (0.2070,0.0660) -- (0.2300,0.0733) -- (0.2530,0.0807) -- (0.2760,0.0880) -- (0.2990,0.0954) -- (0.3220,0.1027) -- (0.3450,0.1101) -- (0.3680,0.1174) -- (0.3910,0.1248) -- (0.4140,0.1322) -- (0.4370,0.1395) -- (0.4600,0.1469) -- (0.4830,0.1543) -- (0.5060,0.1617) -- (0.5290,0.1691) -- (0.5520,0.1765) -- (0.5750,0.1839) -- (0.5980,0.1913) -- (0.6210,0.1987) -- (0.6440,0.2061) -- (0.6670,0.2135) -- (0.6900,0.2209) -- (0.7130,0.2284) -- (0.7360,0.2358) -- (0.7590,0.2433) -- (0.7820,0.2507) -- (0.8050,0.2582) -- (0.8280,0.2657) -- (0.8510,0.2731) -- (0.8740,0.2806) -- (0.8970,0.2881) -- (0.9200,0.2957) -- (0.9430,0.3032) -- (0.9660,0.3107) -- (0.9890,0.3182) -- (1.0120,0.3258) -- (1.0350,0.3333) -- (1.0580,0.3409) -- (1.0810,0.3485) -- (1.1040,0.3561) -- (1.1270,0.3637) -- (1.1500,0.3713) -- (1.1730,0.3789) -- (1.1960,0.3866) -- (1.2190,0.3942) -- (1.2420,0.4019) -- (1.2650,0.4096) -- (1.2880,0.4172) -- (1.3110,0.4250) -- (1.3340,0.4327) -- (1.3570,0.4404) -- (1.3800,0.4481) -- (1.4030,0.4559) -- (1.4260,0.4637) -- (1.4490,0.4715) -- (1.4720,0.4793) -- (1.4950,0.4871) -- (1.5180,0.4949) -- (1.5410,0.5028) -- (1.5640,0.5107) -- (1.5870,0.5185) -- (1.6100,0.5265) -- (1.6330,0.5344) -- (1.6560,0.5423) -- (1.6790,0.5503) -- (1.7020,0.5583) -- (1.7250,0.5662) -- (1.7480,0.5743) -- (1.7710,0.5823) -- (1.7940,0.5903) -- (1.8170,0.5984) -- (1.8400,0.6065) -- (1.8630,0.6146) -- (1.8860,0.6228) -- (1.9090,0.6309) -- (1.9320,0.6391) -- (1.9550,0.6473) -- (1.9780,0.6555) -- (2.0010,0.6638) -- (2.0240,0.6720) -- (2.0470,0.6803) -- (2.0700,0.6886) -- (2.0930,0.6970) -- (2.1160,0.7054) -- (2.1390,0.7137) -- (2.1620,0.7222) -- (2.1850,0.7306) -- (2.2080,0.7391) -- (2.2310,0.7476) -- (2.2540,0.7561) -- (2.2770,0.7646) -- (2.3000,0.7732) -- (2.3230,0.7818) -- (2.3460,0.7904) -- (2.3690,0.7991) -- (2.3920,0.8078) -- (2.4150,0.8165) -- (2.4380,0.8252) -- (2.4610,0.8340) -- (2.4840,0.8428) -- (2.5070,0.8517) -- (2.5300,0.8605) -- (2.5530,0.8694) -- (2.5760,0.8784) -- (2.5990,0.8874) -- (2.6220,0.8964) -- (2.6450,0.9054) -- (2.6680,0.9145) -- (2.6910,0.9236) -- (2.7140,0.9327) -- (2.7370,0.9419) -- (2.7600,0.9511) -- (2.7830,0.9604) -- (2.8060,0.9697) -- (2.8290,0.9790) -- (2.8520,0.9883) -- (2.8750,0.9978) -- (2.8980,1.0072) -- (2.9210,1.0167) -- (2.9440,1.0262) -- (2.9670,1.0358) -- (2.9900,1.0454) -- (3.0130,1.0550) -- (3.0360,1.0647) -- (3.0590,1.0745) -- (3.0820,1.0842) -- (3.1050,1.0941) -- (3.1280,1.1039) -- (3.1510,1.1139) -- (3.1740,1.1238) -- (3.1970,1.1338) -- (3.2200,1.1439) -- (3.2430,1.1540) -- (3.2660,1.1642) -- (3.2890,1.1744) -- (3.3120,1.1846) -- (3.3350,1.1949) -- (3.3580,1.2053) -- (3.3810,1.2157) -- (3.4040,1.2262) -- (3.4270,1.2367) -- (3.4500,1.2473) -- (3.4730,1.2579) -- (3.4960,1.2686) -- (3.5190,1.2793) -- (3.5420,1.2901) -- (3.5650,1.3010) -- (3.5880,1.3119) -- (3.6110,1.3229) -- (3.6340,1.3339) -- (3.6570,1.3450) -- (3.6800,1.3562) -- (3.7030,1.3674) -- (3.7260,1.3787) -- (3.7490,1.3901) -- (3.7720,1.4015) -- (3.7950,1.4130) -- (3.8180,1.4246) -- (3.8410,1.4362) -- (3.8640,1.4479) -- (3.8870,1.4597) -- (3.9100,1.4716) -- (3.9330,1.4835) -- (3.9560,1.4955) -- (3.9790,1.5076) -- (4.0020,1.5197) -- (4.0250,1.5319) -- (4.0480,1.5442) -- (4.0710,1.5566) -- (4.0940,1.5691) -- (4.1170,1.5816) -- (4.1400,1.5943) -- (4.1630,1.6070) -- (4.1860,1.6198) -- (4.2090,1.6327) -- (4.2320,1.6457) -- (4.2550,1.6588) -- (4.2780,1.6719) -- (4.3010,1.6852) -- (4.3240,1.6985) -- (4.3470,1.7120) -- (4.3700,1.7255) -- (4.3930,1.7392) -- (4.4160,1.7529) -- (4.4390,1.7668) -- (4.4620,1.7807) -- (4.4850,1.7948) -- (4.5080,1.8089) -- (4.5310,1.8232) -- (4.5540,1.8376) -- (4.5770,1.8521) -- (4.6000,1.8667) -- (4.6230,1.8814) -- (4.6460,1.8962) -- (4.6690,1.9112) -- (4.6920,1.9263) -- (4.7150,1.9414) -- (4.7380,1.9568) -- (4.7610,1.9722) -- (4.7840,1.9878) -- (4.8070,2.0035) -- (4.8300,2.0193) -- (4.8530,2.0353) -- (4.8760,2.0514) -- (4.8990,2.0677) -- (4.9220,2.0841) -- (4.9450,2.1006) -- (4.9680,2.1173) -- (4.9910,2.1342) -- (5.0140,2.1511) -- (5.0370,2.1683) -- (5.0600,2.1856) -- (5.0830,2.2030) -- (5.1060,2.2207) -- (5.1290,2.2385) -- (5.1520,2.2564) -- (5.1750,2.2745) -- (5.1980,2.2928) -- (5.2210,2.3113) -- (5.2440,2.3300) -- (5.2670,2.3488) -- (5.2900,2.3679) -- (5.3130,2.3871) -- (5.3360,2.4065) -- (5.3590,2.4261) -- (5.3820,2.4459) -- (5.4050,2.4659) -- (5.4280,2.4862) -- (5.4510,2.5066) -- (5.4740,2.5273) -- (5.4970,2.5481) -- (5.5200,2.5692) -- (5.5430,2.5906) -- (5.5660,2.6121) -- (5.5890,2.6340) -- (5.6120,2.6560) -- (5.6350,2.6783) -- (5.6580,2.7008) -- (5.6810,2.7236) -- (5.7040,2.7467) -- (5.7270,2.7701) -- (5.7500,2.7937) -- (5.7730,2.8176) -- (5.7960,2.8417) -- (5.8190,2.8662) -- (5.8420,2.8910) -- (5.8650,2.9160) -- (5.8880,2.9414) -- (5.9110,2.9671) -- (5.9340,2.9931) -- (5.9570,3.0194) -- (5.9800,3.0461) -- (6.0030,3.0731) -- (6.0260,3.1005) -- (6.0490,3.1283) -- (6.0720,3.1564) -- (6.0950,3.1848) -- (6.1180,3.2137) -- (6.1410,3.2430) -- (6.1640,3.2726) -- (6.1870,3.3027) -- (6.2100,3.3332) -- (6.2330,3.3641) -- (6.2560,3.3955) -- (6.2790,3.4273) -- (6.3020,3.4595) -- (6.3250,3.4923) -- (6.3480,3.5255) -- (6.3710,3.5592) -- (6.3940,3.5935) -- (6.4170,3.6282) -- (6.4400,3.6635) -- (6.4630,3.6994) -- (6.4860,3.7358) -- (6.5090,3.7728) -- (6.5320,3.8103) -- (6.5550,3.8485) -- (6.5780,3.8873) -- (6.6010,3.9268) -- (6.6240,3.9669) -- (6.6470,4.0076) -- (6.6700,4.0491) -- (6.6930,4.0913) -- (6.7160,4.1342) -- (6.7390,4.1779) -- (6.7620,4.2223) -- (6.7850,4.2676) -- (6.8080,4.3136) -- (6.8310,4.3605) -- (6.8540,4.4083) -- (6.8770,4.4569) -- (6.9000,4.5065) -- (6.9230,4.5571) -- (6.9460,4.6086) -- (6.9690,4.6611) -- (6.9920,4.7147) -- (7.0150,4.7693) -- (7.0380,4.8250) -- (7.0610,4.8819) -- (7.0840,4.9400) -- (7.1070,4.9993) -- (7.1300,5.0598) -- (7.1530,5.1217) -- (7.1760,5.1849) -- (7.1990,5.2495) -- (7.2220,5.3155) -- (7.2450,5.3830) -- (7.2680,5.4521) -- (7.2910,5.5228) -- (7.3140,5.5951)}
\def\monopolesmallpath{(0.0000,0.0000) -- (0.0230,0.0026) -- (0.0460,0.0051) -- (0.0690,0.0077) -- (0.0920,0.0103) -- (0.1150,0.0128) -- (0.1380,0.0154) -- (0.1610,0.0180) -- (0.1840,0.0205) -- (0.2070,0.0231) -- (0.2300,0.0257) -- (0.2530,0.0282) -- (0.2760,0.0308) -- (0.2990,0.0334) -- (0.3220,0.0360) -- (0.3450,0.0385) -- (0.3680,0.0411) -- (0.3910,0.0437) -- (0.4140,0.0463) -- (0.4370,0.0488) -- (0.4600,0.0514) -- (0.4830,0.0540) -- (0.5060,0.0566) -- (0.5290,0.0592) -- (0.5520,0.0618) -- (0.5750,0.0643) -- (0.5980,0.0669) -- (0.6210,0.0695) -- (0.6440,0.0721) -- (0.6670,0.0747) -- (0.6900,0.0773) -- (0.7130,0.0799) -- (0.7360,0.0825) -- (0.7590,0.0851) -- (0.7820,0.0878) -- (0.8050,0.0904) -- (0.8280,0.0930) -- (0.8510,0.0956) -- (0.8740,0.0982) -- (0.8970,0.1008) -- (0.9200,0.1035) -- (0.9430,0.1061) -- (0.9660,0.1087) -- (0.9890,0.1114) -- (1.0120,0.1140) -- (1.0350,0.1167) -- (1.0580,0.1193) -- (1.0810,0.1220) -- (1.1040,0.1246) -- (1.1270,0.1273) -- (1.1500,0.1300) -- (1.1730,0.1326) -- (1.1960,0.1353) -- (1.2190,0.1380) -- (1.2420,0.1407) -- (1.2650,0.1433) -- (1.2880,0.1460) -- (1.3110,0.1487) -- (1.3340,0.1514) -- (1.3570,0.1541) -- (1.3800,0.1569) -- (1.4030,0.1596) -- (1.4260,0.1623) -- (1.4490,0.1650) -- (1.4720,0.1677) -- (1.4950,0.1705) -- (1.5180,0.1732) -- (1.5410,0.1760) -- (1.5640,0.1787) -- (1.5870,0.1815) -- (1.6100,0.1843) -- (1.6330,0.1870) -- (1.6560,0.1898) -- (1.6790,0.1926) -- (1.7020,0.1954) -- (1.7250,0.1982) -- (1.7480,0.2010) -- (1.7710,0.2038) -- (1.7940,0.2066) -- (1.8170,0.2094) -- (1.8400,0.2123) -- (1.8630,0.2151) -- (1.8860,0.2180) -- (1.9090,0.2208) -- (1.9320,0.2237) -- (1.9550,0.2266) -- (1.9780,0.2294) -- (2.0010,0.2323) -- (2.0240,0.2352) -- (2.0470,0.2381) -- (2.0700,0.2410) -- (2.0930,0.2439) -- (2.1160,0.2469) -- (2.1390,0.2498) -- (2.1620,0.2528) -- (2.1850,0.2557) -- (2.2080,0.2587) -- (2.2310,0.2616) -- (2.2540,0.2646) -- (2.2770,0.2676) -- (2.3000,0.2706) -- (2.3230,0.2736) -- (2.3460,0.2767) -- (2.3690,0.2797) -- (2.3920,0.2827) -- (2.4150,0.2858) -- (2.4380,0.2888) -- (2.4610,0.2919) -- (2.4840,0.2950) -- (2.5070,0.2981) -- (2.5300,0.3012) -- (2.5530,0.3043) -- (2.5760,0.3074) -- (2.5990,0.3106) -- (2.6220,0.3137) -- (2.6450,0.3169) -- (2.6680,0.3201) -- (2.6910,0.3233) -- (2.7140,0.3265) -- (2.7370,0.3297) -- (2.7600,0.3329) -- (2.7830,0.3361) -- (2.8060,0.3394) -- (2.8290,0.3426) -- (2.8520,0.3459) -- (2.8750,0.3492) -- (2.8980,0.3525) -- (2.9210,0.3558) -- (2.9440,0.3592) -- (2.9670,0.3625) -- (2.9900,0.3659) -- (3.0130,0.3693) -- (3.0360,0.3727) -- (3.0590,0.3761) -- (3.0820,0.3795) -- (3.1050,0.3829) -- (3.1280,0.3864) -- (3.1510,0.3899) -- (3.1740,0.3933) -- (3.1970,0.3968) -- (3.2200,0.4004) -- (3.2430,0.4039) -- (3.2660,0.4075) -- (3.2890,0.4110) -- (3.3120,0.4146) -- (3.3350,0.4182) -- (3.3580,0.4219) -- (3.3810,0.4255) -- (3.4040,0.4292) -- (3.4270,0.4328) -- (3.4500,0.4365) -- (3.4730,0.4403) -- (3.4960,0.4440) -- (3.5190,0.4478) -- (3.5420,0.4515) -- (3.5650,0.4553) -- (3.5880,0.4592) -- (3.6110,0.4630) -- (3.6340,0.4669) -- (3.6570,0.4708) -- (3.6800,0.4747) -- (3.7030,0.4786) -- (3.7260,0.4826) -- (3.7490,0.4865) -- (3.7720,0.4905) -- (3.7950,0.4946) -- (3.8180,0.4986) -- (3.8410,0.5027) -- (3.8640,0.5068) -- (3.8870,0.5109) -- (3.9100,0.5150) -- (3.9330,0.5192) -- (3.9560,0.5234) -- (3.9790,0.5276) -- (4.0020,0.5319) -- (4.0250,0.5362) -- (4.0480,0.5405) -- (4.0710,0.5448) -- (4.0940,0.5492) -- (4.1170,0.5536) -- (4.1400,0.5580) -- (4.1630,0.5625) -- (4.1860,0.5669) -- (4.2090,0.5714) -- (4.2320,0.5760) -- (4.2550,0.5806) -- (4.2780,0.5852) -- (4.3010,0.5898) -- (4.3240,0.5945) -- (4.3470,0.5992) -- (4.3700,0.6039) -- (4.3930,0.6087) -- (4.4160,0.6135) -- (4.4390,0.6184) -- (4.4620,0.6232) -- (4.4850,0.6282) -- (4.5080,0.6331) -- (4.5310,0.6381) -- (4.5540,0.6432) -- (4.5770,0.6482) -- (4.6000,0.6533) -- (4.6230,0.6585) -- (4.6460,0.6637) -- (4.6690,0.6689) -- (4.6920,0.6742) -- (4.7150,0.6795) -- (4.7380,0.6849) -- (4.7610,0.6903) -- (4.7840,0.6957) -- (4.8070,0.7012) -- (4.8300,0.7068) -- (4.8530,0.7124) -- (4.8760,0.7180) -- (4.8990,0.7237) -- (4.9220,0.7294) -- (4.9450,0.7352) -- (4.9680,0.7411) -- (4.9910,0.7470) -- (5.0140,0.7529) -- (5.0370,0.7589) -- (5.0600,0.7650) -- (5.0830,0.7711) -- (5.1060,0.7772) -- (5.1290,0.7835) -- (5.1520,0.7897) -- (5.1750,0.7961) -- (5.1980,0.8025) -- (5.2210,0.8090) -- (5.2440,0.8155) -- (5.2670,0.8221) -- (5.2900,0.8287) -- (5.3130,0.8355) -- (5.3360,0.8423) -- (5.3590,0.8491) -- (5.3820,0.8561) -- (5.4050,0.8631) -- (5.4280,0.8702) -- (5.4510,0.8773) -- (5.4740,0.8845) -- (5.4970,0.8919) -- (5.5200,0.8992) -- (5.5430,0.9067) -- (5.5660,0.9143) -- (5.5890,0.9219) -- (5.6120,0.9296) -- (5.6350,0.9374) -- (5.6580,0.9453) -- (5.6810,0.9533) -- (5.7040,0.9614) -- (5.7270,0.9695) -- (5.7500,0.9778) -- (5.7730,0.9861) -- (5.7960,0.9946) -- (5.8190,1.0032) -- (5.8420,1.0118) -- (5.8650,1.0206) -- (5.8880,1.0295) -- (5.9110,1.0385) -- (5.9340,1.0476) -- (5.9570,1.0568) -- (5.9800,1.0661) -- (6.0030,1.0756) -- (6.0260,1.0852) -- (6.0490,1.0949) -- (6.0720,1.1047) -- (6.0950,1.1147) -- (6.1180,1.1248) -- (6.1410,1.1350) -- (6.1640,1.1454) -- (6.1870,1.1559) -- (6.2100,1.1666) -- (6.2330,1.1774) -- (6.2560,1.1884) -- (6.2790,1.1995) -- (6.3020,1.2108) -- (6.3250,1.2223) -- (6.3480,1.2339) -- (6.3710,1.2457) -- (6.3940,1.2577) -- (6.4170,1.2699) -- (6.4400,1.2822) -- (6.4630,1.2948) -- (6.4860,1.3075) -- (6.5090,1.3205) -- (6.5320,1.3336) -- (6.5550,1.3470) -- (6.5780,1.3606) -- (6.6010,1.3744) -- (6.6240,1.3884) -- (6.6470,1.4027) -- (6.6700,1.4172) -- (6.6930,1.4320) -- (6.7160,1.4470) -- (6.7390,1.4623) -- (6.7620,1.4778) -- (6.7850,1.4936) -- (6.8080,1.5098) -- (6.8310,1.5262) -- (6.8540,1.5429) -- (6.8770,1.5599) -- (6.9000,1.5773) -- (6.9230,1.5950) -- (6.9460,1.6130) -- (6.9690,1.6314) -- (6.9920,1.6501) -- (7.0150,1.6693) -- (7.0380,1.6888) -- (7.0610,1.7087) -- (7.0840,1.7290) -- (7.1070,1.7497) -- (7.1300,1.7709) -- (7.1530,1.7926) -- (7.1760,1.8147) -- (7.1990,1.8373) -- (7.2220,1.8604) -- (7.2450,1.8841) -- (7.2680,1.9082) -- (7.2910,1.9330) -- (7.3140,1.9583) -- (7.3370,1.9842) -- (7.3600,2.0108) -- (7.3830,2.0379) -- (7.4060,2.0658) -- (7.4290,2.0944) -- (7.4520,2.1237) -- (7.4750,2.1538) -- (7.4980,2.1846) -- (7.5210,2.2163) -- (7.5440,2.2488) -- (7.5670,2.2822) -- (7.5900,2.3165) -- (7.6130,2.3519) -- (7.6360,2.3882) -- (7.6590,2.4256) -- (7.6820,2.4641) -- (7.7050,2.5037) -- (7.7280,2.5446) -- (7.7510,2.5867) -- (7.7740,2.6302) -- (7.7970,2.6750) -- (7.8200,2.7213) -- (7.8430,2.7692) -- (7.8660,2.8187) -- (7.8890,2.8699) -- (7.9120,2.9228) -- (7.9350,2.9777) -- (7.9580,3.0346) -- (7.9810,3.0936) -- (8.0040,3.1548) -- (8.0270,3.2184) -- (8.0500,3.2845) -- (8.0730,3.3533) -- (8.0960,3.4249) -- (8.1190,3.4995) -- (8.1420,3.5773) -- (8.1650,3.6585) -- (8.1880,3.7434) -- (8.2110,3.8322) -- (8.2340,3.9252) -- (8.2570,4.0227) -- (8.2800,4.1250) -- (8.3030,4.2325) -- (8.3260,4.3456) -- (8.3490,4.4648) -- (8.3720,4.5906) -- (8.3950,4.7235) -- (8.4180,4.8641) -- (8.4410,5.0133) -- (8.4640,5.1717) -- (8.4870,5.3402) -- (8.5100,5.5200)}
\def\figW{9.2000}
\def\figH{5.6000}
\def\axXend{9.5500}
\def\axYend{5.9500}
\def\xmid{4.6000}
\def\xend{9.2000}
\def\yone{1.8667}
\def\ytwo{3.7333}
\def\ythree{5.6000}
\def\labX{4.6000}
\def\labY{2.8000}
\def\legAxA{0.6440}
\def\legBxA{1.5640}
\def\legLxA{1.8400}
\def\legyA{5.3387}
\def\legAxB{0.6440}
\def\legBxB{1.5640}
\def\legLxB{1.8400}
\def\legyB{4.7787}
\def\legAxC{0.6440}
\def\legBxC{1.5640}
\def\legLxC{1.8400}
\def\legyC{4.2187}
\def\thrLx{0.2760}
\def\thrLy{2.1093}

  %% the Randers threshold
  \draw[gray!55,dashed,line width=0.5pt] (0,\yone) -- (\xend,\yone);
  \node[anchor=west,font=\footnotesize,gray!60!black] at (\thrLx,\thrLy)
       {Randers condition};

  %% axes
  \draw[line width=0.6pt] (0,0) -- (\axXend,0);
  \draw[line width=0.6pt] (0,0) -- (0,\axYend);
  \draw[line width=0.6pt] (0,0) -- (0,-0.09) node[below,font=\footnotesize] {$0$};
  \draw[line width=0.6pt] (\xmid,0) -- (\xmid,-0.09)
       node[below,font=\footnotesize] {$\pi/2$};
  \draw[line width=0.6pt] (\xend,0) -- (\xend,-0.09)
       node[below,font=\footnotesize] {$\pi$};
  \draw[line width=0.6pt] (0,\yone)   -- (-0.09,\yone)   node[left,font=\footnotesize] {$1$};
  \draw[line width=0.6pt] (0,\ytwo)   -- (-0.09,\ytwo)   node[left,font=\footnotesize] {$2$};
  \draw[line width=0.6pt] (0,\ythree) -- (-0.09,\ythree) node[left,font=\footnotesize] {$3$};
  \node[below,font=\footnotesize] at (\xmid,-0.62) {colatitude $\theta$};
  \node[rotate=90,font=\footnotesize] at (-0.92,\labY) {$\vert A\vert_g/v$};

  %% the two lengths
  \draw[line width=1.1pt] \monopolepath;
  \draw[line width=1.1pt,densely dashed] \monopolesmallpath;
  \draw[line width=1.1pt,gray!70] \uniformpath;

  %% the equator, where both shape functions take the value one
  \fill (\xmid,\yone) circle (1.5pt);

  %% legend
  \draw[line width=1.1pt] (\legAxA,\legyA) -- (\legBxA,\legyA);
  \node[anchor=west,font=\footnotesize] at (\legLxA,\legyA) {monopole};
  \draw[line width=1.1pt,densely dashed] (\legAxB,\legyB) -- (\legBxB,\legyB);
  \node[anchor=west,font=\footnotesize] at (\legLxB,\legyB)
       {monopole, a third of the charge};
  \draw[line width=1.1pt,gray!70] (\legAxC,\legyC) -- (\legBxC,\legyC);
  \node[anchor=west,font=\footnotesize] at (\legLxC,\legyC)
       {uniform field at $\alpha = v$};
\end{tikzpicture}

  \caption{The Randers condition in the two fields of
  Section~\ref{sec.sphere}. The length of the effective potential, measured
  against the speed, is $(\alpha/v)\sin\theta$ in the uniform ambient
  field~\eqref{eq.spherefield} and $(qk/mv)\tan(\theta/2)$ on a chart of the
  monopole~\eqref{eq.wuyang}. Both shape functions take the value one on the equator, where the curves meet. The first is bounded and reaches the condition
  only at the threshold~\eqref{eq.threshold}, $\alpha = v$, whilst the second
  is unbounded at the pole its chart omits, and no charge however small brings
  it under the condition.}
  \label{fig.randers}
\end{figure}

However, a transport remains, and it is not the one Section~\ref{sec.transport} built. That one is
the Randers spray's own, and the spray is assembled from $F_{\rm R}$, which contains the potential itself --- its middle term needs no primitive, whilst the metric's own spray and the tangential
term ask for one. In its place we have the connection in Ingarden's and Miron's form,
\beq \label{eq.lorentzconn}
N^i{}_j = \mathring{\Gamma}^i{}_{jk}\,y^k - \frac{q}{m}\,g^{ik}F_{kj} ,
\eeq
which is assembled from the field itself. It is defined wherever $F$ is, with the trajectories of the charge for its autoparallel curves. At vanishing charge it returns the Christoffel symbols, as the other connection does. Thus, the distinction we drew there decides the example. Their connection has the field at degree zero, needing no primitive, whilst the one the Randers metric determines is the metric's own, needing one. Their arrangement gives up homogeneity for that freedom, since the connection~\eqref{eq.lorentzconn} fixes a parametrisation where the metric fixed none.

%%------------------------------------------------------------
\subsection{The lattice of charges}\label{sec.lattice}
%%------------------------------------------------------------
The construction has left the multiplier free on a line at every stage. It is constrained nowhere by the geometry, merely made invariant by the flux. One further question is first answered by the monopole. We ask when the kinematic structure a charge selects admits a prequantum line bundle. The answer requires one further input, which we declare here. A unit of action $\hbar$ is a constitutive quantity of the same kind as the mass and the charge, brought to the construction, not derived within it. With it the question above has a classical answer. The structure $(\,T^*\Sph^2, \omega_q)$ admits such a bundle when the class of $\omega_q/2\pi\hbar$ is integral. Since $\omega_0$ is exact that class comes from the twist alone, so with the flux~\eqref{eq.monopolefield} the condition becomes
\beq \label{eq.dirac}
-\frac{q}{2\pi\hbar}\oint_{\Sph^2} F = -\frac{2qk}{\hbar} \in \mathbb{Z} ,
\quad \text{that is} \quad
qk = \frac{n\hbar}{2} ,
\eeq
which is Dirac's condition on the product of an electric and a magnetic charge, with $q$ the charge a particle carries, $k$ the strength of the monopole and $\hbar$ the unit of action we have brought.

The quantisation condition is standard. It follows here from the usual Weil integrality condition together with the Kostant--Souriau prequantisation of the twisted cotangent bundle~\cite{souriau1970,woodhouse1992}, the bundle over the configuration space following by projection. We add what it says about the particle-dependent multiplier the kinematic classification has already left free. Solving the condition~\eqref{eq.dirac} for the charge,
\beq \label{eq.lattice}
q \in \frac{\hbar}{2k}\,\mathbb{Z} ,
\eeq
so that within the prequantisable subclass the affine line of kinematic structures is restricted to a discrete lattice of admissible multipliers. In the lattice~\eqref{eq.lattice} we see in what manner each ingredient enters. Its spacing is fixed by the flux, through the strength $k$, its scale by the unit of action. A charge is no longer free to take any real value once the field it moves in is of non-zero flux. Weaker fields admit a coarser lattice, whilst in the limit $k \to 0$ the spacing grows without bound, which is the flux-free case of Section~\ref{sec.uniform} returning the whole line.

The transport is untouched by any of it. The connection~\eqref{eq.lorentzconn} is assembled from
the metric and the field, both of them defined over the whole sphere, so it is there at every real
charge and the condition~\eqref{eq.dirac} is none of its business. The bundle needed there is the prequantum bundle of the structure, a different bundle over a
different space. Thus, the
flux constrains the two halves of the pair unequally. It forbids the form outright, and leaves the transport in place, unquantised. That is as far as the construction goes, further than the geometry alone could. The step required a unit we had to bring.

%%============================================================
\section{Discussion}\label{sec.spacetime}
%%============================================================

\emph{What must a structure do to be called spacetime?} We have found, for a charge, the pair inertial motion requires --- a form together with a law of transport --- one level up, on the bundle, and have seen a flux part them. Whether a structure depending on the particle which moves in it is a spacetime is a question we put to a criterion of universality, taken from someone who did not have our example in view.

Knox proposes that the spacetime role is played by whatever defines a structure of local inertial
frames, and applies it to Newtonian gravitation in two steps~\cite{knox2011,knox2014}. The split
between gravity and inertia is not observable, no experiment distinguishing a homogeneous field
from a uniform acceleration, so only the sum of the two is a candidate for structure. And her
Newtonian Strong Equivalence Principle asks that in a small enough region there be a freely
falling frame in whose coordinates the motions of bodies amongst themselves are those of bodies
free of external forces, which identifies the freely falling frames with the inertial ones and
makes the connection of geometrised Newtonian gravitation the one representing spacetime
structure. Both steps are hers, taken here unchanged.

Our construction offers the same two for a charge. Gauge freedom in the potential is a symmetry of
the kinematics, and a charge runs along the autoparallel curves of the nonlinear connection
assembled from the field, so that a frame adapted to that structure is to a charge what
a freely falling frame is to a massive body. Were the criterion to answer alike in both cases, the
geometry built here would be the spacetime of a charged particle.

However, it does not answer alike. Knox's principle asks after \emph{the motions of bodies amongst themselves}, quantifying over all bodies at once. That quantifier our structures cannot satisfy. A frame in which one value of the ratio moves inertially is a frame in which every other value moves under a force. In the structure attached to one value of $q/m$, a particle of another value $q'/m'$ at the same speed departs from inertial motion by
\beq \label{eq.mismatch}
\left( \frac{q}{m} - \frac{q'}{m'} \right)\iiota_{\dot\gamma}F .
\eeq
Linear in the mismatch of the two ratios, the failure vanishes with it. That it vanishes for a homogeneous field is de Matos, {\"O}zer and Izworski's observation~\cite{dematos2018}. What follows below needs only the departure~\eqref{eq.mismatch}, which is proved here. For a single species the criterion is met, that species having its own structure of local inertial
frames, and the failure is across species. \emph{The equivalence principle is the collapse of that
line. Knox's condition is the collapse read as a criterion.}

Thus, we reach a conclusion Knox draws herself, that not every connection represents spacetime structure. Our construction exhibits a connection defining a structure of local inertial frames for one species of particle --- an effective inertial structure by her criterion of universality. The obstruction to calling it a spacetime is its dependence on the ratio a particle carries. That dependence is the failure of the line to collapse. Whether some other construction geometrises electromagnetism is a question we neither ask nor settle.

Let us grant the objection that follows, since it is our own conclusion arriving from the other side. Read and Menon argue that inertial frame structure comes apart from the structure surveyed by rods and clocks built from matter fields, a gap which Knox's functionalism inherits~\cite{readmenon2019}. Indeed, a charge-dependent structure is surveyed by no neutral rod, nor by any single instrument. We take that as decisive and have called these structures effective throughout.

There remains the quantity by which the species is labelled, upon which the construction says something the criterion cannot. That ratio is a constitutive quantity brought to the theory, while the classification places it --- in the multiplier the geometry leaves free. Inertial motion leaves the scale of the bilinear form and the antisymmetric part of the transport unresolved~\cite{lm2026newton}, the first of them settled by a particle as its mass, the second as a candidate for its spin. \emph{Each constitutive quantity is a freedom the inertial structure cannot resolve.} It holds over three instances. We would put the same question to any quantity a particle is said to bring.

Let us state as much of that as has been proved. The affine reparametrisations of
Section~\ref{sec.ratio} leave the quotient fixed, so a trajectory keeps the mass and the charge
only through the ratio $q/m$, and each of its values selects one member of the twisted
family~\eqref{eq.family} up to those reparametrisations --- by the Hilbert-form
identification~\eqref{eq.hilbert}, one Randers metric on the free level whose geodesics are the
trajectories. Were one ratio common to every body, all of them would select the same member and
inertial motion would be the geodesic flow of one metric. A line which does not collapse gives
each species a structure of its own, where a collapsed line gives one to every body at once.

The manner of the failure decides the local-frame requirement.
The departure~\eqref{eq.mismatch} is the field contracted with the velocity at the point,
multiplied by the mismatch of the ratios, with neither a derivative of the field in it nor any
separation between the two motions. It therefore does not diminish as the two are brought together. A tidal
term does diminish, the relative acceleration of two motions of one family vanishing with their
separation, which is what leaves a local inertial frame at every point of a gravitational field. Here no local
frame is free for both species at once wherever the field does not vanish, and there is the whole
of the difference between an arena and a structure attached to one species.

That each ratio has a geometry of its own is {\"O}zer's, who postulates an equivalence principle for electromagnetism and builds the metric which makes the Lorentz motion geodesic~\cite{ozer1999}. Here the family is derived from the Legendre demand and the metric arrives on the free level, the order of the two being the whole of the difference.

An arena common to every body --- a spacetime, on her criterion --- comes from the collapse of that line, which is to say by an equivalence principle. Such a principle is not about a force. It
asserts that one ratio serves every body, and a single family of inertial motions is thereby left,
from whose relative deviation a geometry is obtained. Universality is therefore contingent,
resting on the measured equality of the gravitational and inertial masses. Electromagnetism is the
case in which it fails.

The counterfactual is de Matos, {\"O}zer and Izworski's, who put an equivalence principle to each
family of particles sharing a ratio and record that ours is the universe with a single
family~\cite{dematos2018}. For them it is an extension of the principle, where for us it is the
reason no spacetime follows. Peres deflated the principle earlier and further, holding that rods
and clocks are macroscopic bodies bound by the field whose geometry they are used to measure, so
that an equivalence principle is no law of Nature but a law of the language in which Nature is
described~\cite{peres1962}. We put the contingency elsewhere, on a measured equality of masses and
on a species-dependence the classification exhibits. In gravitation the two masses coincide, so one
ratio serves every body, and a gravitational force, being of degree zero in the velocity, enters
the energy where our field twists the form, leaving the single affine connection in which Knox
finds spacetime structure. Here the charge and the mass are independent, and our line keeps every
member. \emph{An equivalence principle leaves a single family of inertial motions, from whose
deviation a geometry is obtained. Where the families are many, no geometry is common to them.}

%%============================================================
\section{Final remarks}\label{sec.closing}
%%============================================================

A field adds no force to the motion of a charge. It changes the geometry with respect to which that motion is inertial. The charge-to-mass ratio gives the amount of the change. On the phase space that geometry is fixed entirely, with one multiplier left free within it. On the configuration space nothing of the kind is fixed, since no affine connection serves there, whilst one level up, on the bundle, we obtain a metric together with a law of transport. Only the transport remains once the flux admits no primitive. Which symplectic structures are kinematic we did not choose (Theorem~\ref{theo.kin}). Neither did we choose where a field may enter Newton's Second Law through a constitutive morphism (Proposition~\ref{prop.degreetwo}). For the rest we followed a route, with two constitutive declarations along it --- the posted equality~\eqref{eq.constitutive} for the charge, and the posted equality~\eqref{eq.postmomentum} for the momentum of a motion.

Thus, one account locates both of the failures behind the construction. A nonsymmetric metric
was tried and gave no Lorentz force~\cite{einstein1945,schrodinger1950,goenner2014}. An affine
connection was tried next and did no better~\cite{barros2005gauss}. Each had been argued on its own
terms, and they fail at neighbouring places in it. The antisymmetric part of a
bilinear morphism contributes a term of degree two which the congruence fixes and the velocity at a
point does not, where a Lorentz force is of degree one and pointwise
(Proposition~\ref{prop.degreetwo}). The Randers spray contributes a term of degree two as well, as
every spray must, and that term is pointwise, failing only to be quadratic in the velocity --- which
is the freedom an affine connection lacks and the one Section~\ref{sec.transport} takes up. Neither
failure is the other, and the decomposition~\eqref{eq.spray} together with
Proposition~\ref{prop.degreetwo} says where each of them falls.

The metric which serves in place of the connection was taken from the structure the flow already
occupies. On the free level the primitive of the twisted form is $v$ times the Hilbert form of a
Randers metric, put there by the Liouville field, a dilation centred at the potential
(Section~\ref{sec.finsler}). Read as a length, that metric answers differently to an indefinite
signature. Cauchy--Schwarz reverts on the timelike cone, the condition on the potential becomes a
lower bound which grows with it, and the trajectories are the longest curves of a length Randers built for the shortest (Section~\ref{sec.indefinite}). The dynamics is indifferent to
all of that, needing only that the Randers function not vanish, so the length and the law come
apart where the metric is indefinite.

The order in which we admitted those structures shows its value only from the end. Nothing in the metric or in its law of transport was assumed before the motion needed it. A reader who does not grant a step may therefore see how much of the construction depends on it, and how much is given without it. That is what the order gives us.

Electromagnetism gives a geometry to every value of the charge-to-mass ratio, one too many for a single arena. In Section~\ref{sec.spacetime} an equivalence principle asserts the collapse of that line to a single member, as gravitation does once its two masses are found equal by measurement. Universality is therefore a property of the collapse itself.

Section~\ref{sec.lattice} is the one place where this construction fixes a number. The condition is that the structure admit a prequantum bundle. It falls on the charge, since the charge multiplies the class of the field in cohomology. A mass is only ever a scale, and a scale admits no lattice. Klein's fifth dimension gives the same conclusion from the other side~\cite{klein1926}, where a mass --- a momentum around a compact direction --- is quantised by the periodicity which quantises the charge. Change the freedom and the answer changes. In this sense the question is one about the freedom a quantity ranges over.

The mass and the spin have been treated in this manner before~\cite{lm2026newton}. The charge is the third. Each is a freedom which inertial motion leaves open and which a particle closes on its own account. The charge is the one of the three which the geometry closes in part as well, once a prequantum bundle is needed. Whether that belongs to electromagnetism alone, or is the first instance of a wider rule, we raise here without settling.

Let us set out the matters we leave open. The first is the multiplier, with which our own method is
least at ease. The rank analysis of Section~\ref{sec.charge} settles what the
classification asks of the charge, which above rank four is constancy and on a surface is nothing
at all. A charge varying from place to place is therefore left standing by the geometry. What excludes
it comes from outside, in the invariance requirement we import~\cite{lm2026newton} and in our own
reading of the multiplier as the charge of one particle. A manuscript which admits each
structure at the point of need, and marks each declaration where it is made, would rather have
derived that one. Settling it from within would take a principle of invariance for inertial motion
stronger than the one we have, and we do not have it.

A second question the construction raises and does not answer concerns the two-form itself. Feynman's
argument was carried to a non-Abelian gauge theory~\cite{tanimura1992}, where the twist would take
its values in the adjoint bundle rather than in the reals, and Sternberg's reduction of a principal
bundle~\cite{sternberg1977,guillemin1990} is where such a twist would already have its setting. Whether the Legendre demand, asked of every energy at once, classifies the structures there
as it does here we have not examined. The question is worth putting precisely, and we put it
without pursuing it. The curvature of the metrics we obtain, given explicitly for the round sphere in a uniform field, we leave to a manuscript of its own. The other half of Maxwell's equations we have not entered, but it would be remiss not to say what would be needed for them. The excitation and the current need a metric together with a second constitutive declaration, and a spacetime where we have posed a configuration manifold alone.

The order in which we demanded the pieces of this construction is ours. It shows a reader, at every step, how much of the answer was already given before a charge was mentioned at all. Fields are inputs throughout, and every trajectory in the manuscript is one the Lorentz force law already gives. In sum, we conclude that the motion of a charged particle fixes the class of its kinematic structures and the degree at which a field may enter its law of motion, whilst the geometry we obtain from it is one a body shares only with those of its own charge-to-mass ratio. For the first two the mathematics leaves us no choice. For the third it leaves us one, and that is why this construction ends without a spacetime.

%%============================================================
%%  Bibliography
%%============================================================
\bibliographystyle{amsplain}
\bibliography{kinematics}

\end{document}